\documentclass{amsart}
\usepackage{enumerate}
\usepackage{url}
\usepackage{textcomp}
\usepackage{rotating}
\usepackage{multirow}
\usepackage{amsmath}
\usepackage{amssymb}

\def\mb{\mathbf}
\def\Ra{\Rightarrow}
\def\La{\Leftarrow}
\def\ra{\rightarrow}
\def\l{\langle}
\def\r{\rangle}
\def\mb{\mathbf}

\newtheorem{example}{\bf Example}[section]
\newtheorem{lemma}{\bf Lemma}[section]
\newtheorem{theorem}{\bf Theorem}[section]
\newtheorem{proposition}{\bf Proposition}[section]
\newtheorem{definition}{\bf Definition}[section]
\newtheorem{remark}{\bf Remark}[section]

\newtheorem{corollary}{\bf Corollary}[section]

\newtheorem{fact}{\bf Fact}[section]

\begin{document}

\markboth{Sl. Shtrakov and I. Damyanov}
{On the complexity of finite valued functions}

\title{On the complexity of finite valued functions}

\author{Slavcho Shtrakov}

\address{Department of Computer Science,\\ South-West University, Blagoevgrad, Bulgaria\\
%\email{shtrakov@swu.bg}
}

\author{Ivo Damyanov}

\address{Department of Computer Science,\\ South-West University, Blagoevgrad, Bulgaria\\
%\email{damianov@swu.bg}
}

\begin{abstract}
The essential variables in a finite function $f$ are defined as variables which occur in $f$ and weigh with the values of that function.
  The number of   essential  variables is an important measure of complexity for discrete functions.
   When replacing some variables in a function with constants  the resulting functions are called subfunctions,  and when replacing all essential variables in a function with constants we obtain an implementation of this function.
 Such an implementation corresponds with a path in an ordered decision diagram (ODD) of the function which connects the root with a leaf of the diagram.     The sets of essential variables in subfunctions of $f$ are called separable in $f$.
In this paper we study several properties of separable sets of variables in functions which directly impact on the number of implementations and  subfunctions in these functions.

We define equivalence relations which classify the functions of $k$-valued logic into classes with same number of implementations, subfunctions or separable sets. These relations induce three transformation groups which are compared with the lattice of all subgroups of restricted affine group (RAG). This allows us to solve several important computational and combinatorial problems. 

\end{abstract}

\maketitle
\section{Introduction}\label{sec1}

Understanding the complexity of $k$-valued functions is still one of the fundamental tasks
in the theory of computation. At present, besides classical methods like substitution
or degree arguments a bunch of combinatorial and algebraic techniques have been
introduced to tackle this extremely difficult problem. There has been significant progress analysing the power of
randomness and quantum bits or multiparty communication protocols that help to
capture the complexity of switching  functions. For tight estimations concerning the
basic, most simple model switching circuits there still seems a long way to go  (see \cite{comp_sem}).

In Section \ref{sec2} we introduce the basic notation and give definitions of  separable sets, subfunctions,  etc.
The properties of distributive sets of variables with their s-systems are also discussed  in Section \ref{sec2}. In Section \ref{sec3} we study the ordered decomposition trees (ODTs),  ordered decision diagrams (ODDs), and implementations of discrete functions.
We also discuss several problems with the complexity of representations of functions with respect to their ODDs, subfunctions and separable sets. In Section \ref{sec5} we  classify discrete functions by transformation groups and equivalence relations concerning the number of implementations, subfunctions and separable sets in functions. In  Section \ref{sec6} we  use these results  to classify all boolean (switching) functions with "small" number of its essential variables. Here we calculate the number of equivalence classes and cardinalities of equivalence classes of boolean functions depending on 3, 4 and 5 variables. 
\section{Separable and Distributive Sets of Variables}\label{sec2}

We start this section with basic definitions and notation.
A  discrete function is defined as a mapping: $f:A\to B$ where the domain $A={\times}_{i=1}^n A_i$                                                                                                                                          and range $B$ are non-empty finite or countable sets.

To derive the means and methods to represent, and calculate with finite valued functions, some algebraic structure is imposed on the domain $A$ and the range $B$. Both $A$ and $B$ throughout the present paper will be finite ring   of integers.

  Let $X=\{x_1,x_2,\ldots \}$ be a countable set of variables and $X_n=\{x_1,x_2,\ldots,x_n\}$ be a finite subset of $X$. Let $k$,  $k\geq 2$ be a natural number
 and let us denote by $Z_k=\{0,1,\ldots,k-1\}$ the set (ring) of
 remainders modulo $k$. The set $Z_k$ will identify the ring of residue classes $mod\ k$, i.e. $Z_k=Z/_{kZ}$, where $Z$ is the ring of all integers. An {\it $n$-ary $k$-valued  function
(operation) on $Z_k$ } is a mapping $f: Z_k^n\to Z_k$ for some natural
number $n$, called {\it the arity}  of $f$. $P_k^n$ denotes the set of all $n$-ary $k$-valued functions on $Z_k$. It is well known fact that there are $k^{k^n}$ functions in $P_k^n$. The set of all $k$-valued
functions
$P_k=\bigcup_{n=1}^\infty P_k^n$
is called {\it the algebra of $k$-valued logic}.

 All results  obtained in the present paper can be easily extended to arbitrary algebra of finite operations.

For a given  variable $x$ and $\alpha\in Z_k$, $x^\alpha$ is defined as follows:
$$
  x^\alpha=\left\{\begin{array}{ccc}
             1 \  &\  if \  &\  x=\alpha \\
             0 & if & x\neq\alpha.
           \end{array}
           \right.
$$

We   use {\it sums of
products (SP)} to  represent the functions from $P_k^n$. This is the most natural representation and it is based on so called operation tables of the functions.
Thus each function $f\in P_k^n$ can be uniquely represented of SP-form as
follows
\[    f=a_0.x_1^{0}\ldots x_n^{0}\oplus\ldots\oplus
     a_{m}.x_1^{\alpha_1}\ldots
     x_n^{\alpha_n}\oplus\ldots\oplus a_{k^n-1}.x_1^{k-1}\ldots
     x_n^{k-1}\]
with  ${\mathbf{\alpha}}={({\alpha_1,\ldots,\alpha_n})}\in Z_k^n$,  where  $m=\sum_{i=0}^n\alpha_ik^i\leq k^n-1$.  $"\oplus"$ and
$"."$ denote the operations addition (sum) and multiplication (product) modulo $k$ in
the ring $Z_k$. Then $(a_0,\ldots,a_{k^n-1})$ is the vector of output values of $f$ in its table representation.

  Let $f\in P_k^n$ and   $var(f)=\{x_1,\ldots,x_n\}$  be the set of all variables, which occur in $f$.
   We say that the $i$-th variable $x_i\in var(f)$ is  {\it  essential}  in $f$, or $f$ {\it
essentially} {\it depends} on $x_i$, if there exist values
$a_1,\ldots,a_n,b\in Z_k$, such that
\[
   f(a_1,\ldots,a_{i-1},a_{i},a_{i+1},\ldots,a_n)\neq  f(a_1,\ldots,a_{i-1},b,a_{i+1},\ldots,a_n).
\]

The set of all essential variables in the function $f$ is denoted by
$Ess(f)$ and the number of  essential variables  in $f$ is denoted by
$ess(f)=|Ess(f)|$.
 The variables from $var(f)$ which
are not essential in
 $f$ are called {\it inessential} or {\it fictive}.

 The set of all output values of a function $f\in P_k^n$ is called {\it the range} of $f$, which is denoted as follows:
 \[range(f)=\{c\in Z_k\ |\ \exists (a_1,\ldots,a_n)\in Z_k^n, \quad such\ that\quad f(a_1,\ldots,a_n)=c\}.\]
\begin{definition}\label{d1.2}
Let $x_i$  be an essential variable in $f$ and $c\in Z_k$ be a constant from $Z_k$. The
function $g=f(x_i=c)$ obtained from $f\in P_{k}^{n}$ by replacing the variable $x_i$ with $c$  is
called a {\it simple subfunction of $f$}.

When $g$ is a simple subfunction of $f$ we shall write $g\prec f$. The transitive closure of $\prec$ is denoted by $\preceq$.   $Sub(f)=\{g\ | \ g\preceq f\}$ is the set of all subfunctions of $f$ and  $sub(f)=|Sub(f)|$.
\end{definition}
 For each $m=0,1,\ldots, n$ we denote by  $sub_m(f)$  the number of subfunctions in $f$ with $m$ essential variables, i.e. $sub_m(f)=|\{g\in Sub(f)\ |\ ess(g)=m\}|$.

Let $g\preceq f$, $\mathbf{c}=(c_1,\ldots,c_m)\in Z_k^m$ and $M=\{x_1,\ldots,x_m\}\subset X$ with \[g\prec g_1\prec\ldots\prec g_m=f,\quad
g=g_1(x_1=c_1)\quad and \quad g_i=g_{i+1}(x_{i+1}=c_{i+1})\]
for $i=1,\ldots,m-1$. Then we shall write $g=f(x_1=c_1,\ldots,x_m=c_m)$ or equivalently,  $g\preceq_M^{\mathbf{c}} f$  and we shall say that the vector $\mathbf{c}$ {\it determines} the subfunction $g$ in $f$.

\begin{definition}\label{d1.3}
Let $M$  be a non-empty set of essential variables in the function $f$.
Then   $M$ is called {\it a separable set} in $f$ if there exists a subfunction $g$, $g\preceq f$ such that $M=Ess(g)$.
   $Sep(f)$ denotes the set of the all separable sets in $f$ and  $sep(f)=|Sep(f)|$.
\end{definition}

The sets of essential variables in $f$ which are not separable are called {\em inseparable} or {\em non-separable}.

For each $m= 1,\ldots, n$ we denote by  $sep_m(f)$  the number of separable sets in $f$ which consist of $m$ essential variables, i.e. $sep_m(f)=|\{M\in Sep(f)\ |\ |M|=m\}|$.
The numbers $sub(f)$ and $sep(f)$ characterize the computational complexity of the function $f$ when calculating its values. Our next goal is to classify the functions from $P_k^n$ under these complexities.
The initial and more fundamental results concerning essential variables and separable sets were obtained in the work of Y. Breitbart \cite{bre},  K. Chimev \cite{ch51}, O. Lupanov \cite{lup}, A. Salomaa \cite{sal}, and others.
~

\begin{remark}\label{r1.1} Note that if $g\preceq f$ and $x_i\notin Ess(f)$ then $x_i\notin Ess(g)$ and also if $M\notin Sep(f)$ then  $M\notin Sep(g)$.
\end{remark}

\begin{definition}\label{d2.1} Let $M$ and  $J$ be two non-empty sets of essential variables in the function $f$ such that $M\cap J=\emptyset.$
The set  $J$ is called {\it a
distributive set of  $M$ in $f$,
} if for every $|J|$-tuple of constants $\mathbf{c}$ from $Z_k$ it holds $M\not\subseteq Ess(g)$, where  $g\preceq_J^{\mathbf{c}} f$
 and $J$ is minimal  with respect to the order $\subseteq$.
  $Dis(M,f)$ denotes the set of the all distributive sets of $M$ in $f.$
\end{definition}

 It is clear that if $M\notin Sep(f)$ then $Dis(M,f)\neq\emptyset$. So, the distributive sets ``dominate'' on the inseparable sets of variables in a function. We are interested in the relationship between the structure of the distributive sets of variables and complexity of functions concerning $sep(f)$ and $sub(f)$, respectively, which is illustrated by the following example.

 \begin{example}\label{ex2.1}
  Let $k=2$,  $f=x_1x_2\oplus x_1x_3$ and $g=x_1x_2\oplus x_1^0x_3.$
It is easy to verify that the all three pairs of variables $\{x_1,x_2\}$, $\{x_1,x_3\}$ and $\{x_2,x_3\}$ are separable in $f$, but  $\{x_2,x_3\}$ is inseparable in $g$. {Figure} \ref{f1} presents graphically, separable pairs in $f$ and $g$.
The set $\{x_1\}$ is distributive of $\{x_2,x_3\}$ in $g$.
\end{example}
\begin{figure}[h]
\begin{center}
\setlength{\unitlength}{1cm}
%\special{em:linewidth 0.4pt}
%\linethickness{0.4pt}
\begin{picture}(12.00,4.00)
\thicklines
\put(1.00,2.00){\line(1,0){3.00}}
\put(1.00,2.00){\line(1,1){1.50}}
\put(4.00,2.00){\line(-1,1){1.50}}
\put(0.50,2.00){\makebox(0,0)[cc]{$x_2$}}
\put(2.50,4.00){\makebox(0,0)[cc]{$x_1$}}
\put(4.50,2.00){\makebox(0,0)[cc]{$x_3$}}

\put(6.50,2.00){\makebox(0,0)[cc]{$x_2$}}
\put(8.50,4.00){\makebox(0,0)[cc]{$x_1$}}
\put(10.50,2.00){\makebox(0,0)[cc]{$x_3$}}

\put(2.50,0.50){\makebox(0,0)[cc]{$f=x_1 x_2\oplus x_1x_3$}}
\put(8.50,0.50){\makebox(0,0)[cc]{$g=x_1x_2\oplus x_1^0x_3$}}

\put(2.50,3.50){\circle*{0.2}}
\put(4.00,2.00){\circle*{0.2}}
\put(1.00,2.00){\circle*{0.2}}

\put(8.50,3.50){\circle*{0.2}}
\put(10.00,2.00){\circle*{0.2}}
\put(7.00,2.00){\circle*{0.2}}

%\put(1.00,2.00){\line(1,0){3.00}}
\put(7.00,2.00){\line(1,1){1.50}}
\put(10.00,2.00){\line(-1,1){1.50}}
\end{picture}
\end{center}
 \caption{Separable pairs.}\label{f1}
\end{figure}
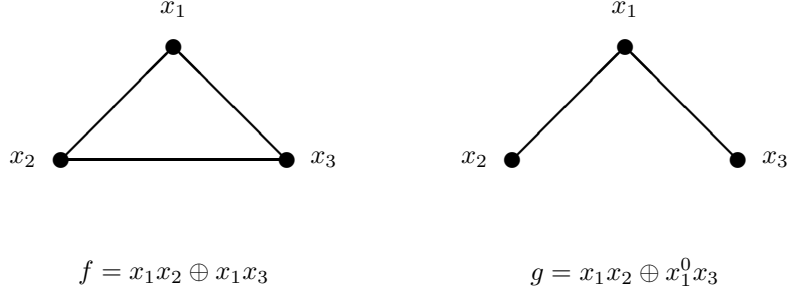

\begin{definition}\label{d2.2}
Let $\mathcal F=\{P_1,\ldots,P_m\}$ be a family of non-empty sets. A set  $\beta=\{x_1,\ldots,x_p\}$ is called
{\it an $s$-system} of  $\mathcal F$,  if for all
$P_i\in \mathcal F$, $1\leq i\leq m$ there exists  $x_j\in\beta$ such that
$x_j\in P_i$ and for all  $x_s\in\beta$ there exists  $P_l\in
\mathcal F$ such that $\{x_s\}=P_l\cap\beta.$
$Sys(\mathcal F)$ denotes the set of the all $s$-systems of the family $\mathcal F$.
\end{definition}

Applying the results concerning  $s$-systems of distributive sets is one of the basic tools for studying inseparable pairs and inseparable sets in functions. These results are deeply discussed in \cite{ch51,s27,s23}.

\begin{theorem}\label{t2.1}\ Let $M\subseteq Ess(f)$ be a  non-empty inseparable set  of essential variables in  $f\in P_k^n$  and $\beta\in Sys(Dis(M,f))$. Then the following statements  hold:
\begin{enumerate}
 \item[(i)] $M\cup \beta \in Sep(f)$;
 \item[(ii)] $(\ \forall \alpha,\ \alpha\subseteq\beta,\ \alpha\neq\beta)\quad  M\cup \alpha\notin Sep(f)$.
\end{enumerate}
\end{theorem}
\begin{proof}
$(i)$\
First, note that $M\notin Sep(f)$ implies $|M|\geq 2$. Without loss of generality assume that
$\beta=\{x_1,\ldots,x_m\}\in Sys(Dis(M,f))$ and $M=\{x_{m+1},\ldots,x_{p}\}$
with $1\leq m<p\leq n$.
Let us denote by $L$ the following set of variables
$L=Ess(f)\setminus(M\cup\beta )=\{x_{p+1},\ldots,x_{n}\}.$

Since $\beta\in Sys(Dis(M,f))$ it follows that for each $Q\subseteq L$ we have $Q\notin Dis(M,f)$. Hence  there is a vector $\mathbf{c}=(c_{p+1},\ldots,c_n)\in Z_k^{n-p}$ such that $M\subseteq Ess(g)$ where $g\preceq_L^{\mathbf{c}} f$.

Next, we shall prove that
 $\beta\subset Ess(g).$
 For suppose this were not the case and without loss of generality, assume
$x_1\notin Ess(g)$, i.e. $g=g(x_1=d_1)$ for each
$d_1\in Z_k.$ Let  $J\in Dis(M,f)$ be a distributive set of $M$ such that
$J\cap\beta=\{x_1\}.$
The existence of the set $J$ follows because $\beta$ is an $s$-system of $Dis(M,f)$ (see Definition \ref{d2.2}). 
Since $J\cap M=\emptyset$ and
 $x_1\notin
Ess(g)$ it follows that  $J\cap Ess(g)=\emptyset.$
Now $M\subset Ess(g)$ implies that
$J\notin Dis(M,f)$, which is a contradiction. Thus we have  $x_1\in Ess(g)$ and $\beta\subset Ess(g)$.
Then
$Ess(f)\setminus L=M\cup\beta$
shows that $M\cup\beta=Ess(g)$ and hence $M\cup\beta\in
Sep(f).$

$(ii)$\ Let $\alpha$, $\alpha\subseteq\beta,\ \alpha\neq\beta$ be a proper subset of $\beta$. 
Let $x_i\in \beta\setminus \alpha$. Then $\beta\in Sys(Dis(M,f))$ implies that there 
 is a distributive set $P\in Dis(M,f)$ of $M$ such that $P\cap\beta=\{x_i\}$. Hence 
$P\cap\alpha=\emptyset$
which shows that there is an non-empty distributive set $P_1$ for $M\cup\{\alpha\}$ with  $P_1\subseteq P$. Hence  
   $M\cup\alpha\notin Sep(f)$.
 \end{proof}
\begin{corollary}\label{c2.1}
Let $\emptyset\neq M\subset Ess(f)$ and $M\notin Sep(f)$. If $\beta\in Sys(Dis(M,f))$ and $x_i\in\beta$ then
$M\setminus Ess(f(x_i=c))\neq\emptyset$ for all $c\in Z_k$.
\end{corollary}

\begin{theorem}\label{t2.2}\cite{s23}
For each finite family $\mathcal F$ of non-empty sets there exists at least one
$s-$system of $\mathcal F$.
\end{theorem}

\begin{theorem}\label{t2.3}
Let $M$ be an inseparable set in $f$. A set $\beta\subset Ess(f)$ is an $s$-system of $Dis(M,f)$ if and only if $\beta\cap J\neq\emptyset$ for all $J\in Dis(M,f)$ and $\alpha\subseteq\beta,\ \alpha\neq\beta$ implies $\alpha\cap P=\emptyset$ for some $P\in Dis(M,f)$.

\end{theorem}
\begin{proof}
"$\La$" Let $\beta\cap J\neq \emptyset$ for all $J\in Dis(M,f)$ and $\alpha\varsubsetneqq \beta$ implies $\alpha\cap P=\emptyset$ for some $P\in Dis(M,f)$. Since $\beta\cap J\neq \emptyset$ it follows that there is a set $\beta'$, $\beta'\subset \beta\subset Ess(f)$ and  $\beta'\in Sys(Dis(M,f))$. If we suppose that $\beta'\neq\beta$ then there is $P\in Dis(M,f)$ with $\beta'\cap P=\emptyset$.  Hence $M\cup\beta'\notin Sep(f)$ because of $P\in Dis(M\cup\beta',f)$ which contradicts Theorem \ref{t2.1}.

"$\Ra$" Let $\beta$ be an $s$-system of $Dis(M,f)$ and $\alpha\varsubsetneqq \beta$. Let $x\in \beta\setminus\alpha$ and $P\in Dis(M,f)$ be a distributive set of $M$ for which $\beta\cap P=\{x\}$. Hence $\alpha\cap P=\emptyset$ and  we have $P\in Dis(M\cup \alpha,f)$ and  $M\cup\alpha\notin Sep(f)$ which shows that $\alpha\notin Sys(Dis(M,f))$.
~~~~~

\end{proof}

\section{Ordered Decision Diagrams and Complexity of Functions}\label{sec3}
%%%%%%%%%%%%%%%%%%%%%%%%%%%%%%%%%%%%%%%%%%

The distributive sets are also important  when constructing efficient procedures for simplifying in analysis and synthesis of functional schemas. 

  In this section we discuss  {\it ordered decision diagrams} (ODDs) for the functions obtained by restrictions on their {\it ordered decomposition trees} (ODTs).

 {Figure} \ref{f2}  shows an ordered decomposition tree for the function $g=x_1x_2\oplus x_1^0x_3\in P_2^3$ from  Example \ref{ex2.1},  which essentially depends on all its three variables $x_1, x_2$ and $x_3$.
The node at the top,  labelled  $g$ - is the {\it function} node.
The nodes drawn as filled circles labelled with variable names are the {\it internal (non-terminal)} nodes, and the rectangular nodes (leaves of the tree) are the {\it terminal} nodes. The terminal nodes are labelled by the  numbers from $Z_k$. Implementation of $g$ for a given values of $x_1, x_2$ and $x_3$ consists of selecting a path from the function node to a terminal node. The label of the terminal node is the sought value. At each non-terminal node the path follows the solid edge if the variable labelling the node evaluates to $1$, and the dashed edge if the variable evaluates to $0$. In the case of $k>2$ we can use colored edges with $k$ distinct colors. 

The ordering in which the variables appear is  the same along all paths of an ODT.   {Figure} \ref{f2} shows the ODT for the function $g$ from Example \ref{ex2.1}, corresponding to the variable ordering  $x_1,x_2,x_3$ (denoted briefly as $\l 1; 2; 3\r$). It is known that for a given function $g$ and a given ordering of its essential variables there is a unique ODT.

 We extend our study to ordered decision diagrams for the functions from $P_k^n$ which were  studied by D. Miller and R. Drechsler \cite{mil1,mil2}.   
 %%%%%%%%%%%%%%%%%%%%%%%%%%%%%%%
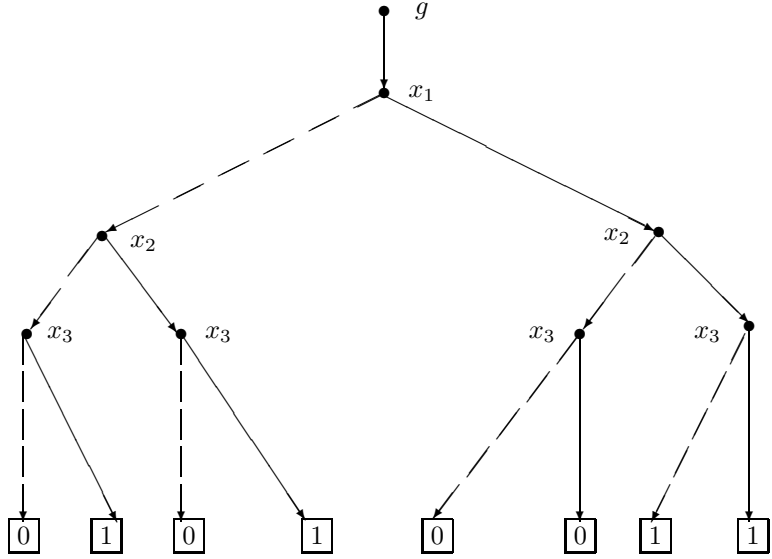
\begin{figure}%[b]
\centering

\unitlength 1mm % = 2.845pt
\linethickness{0.4pt}
\begin{picture}(115,75)(10,0)
\put(65,73){\circle*{1.5}}
\put(70,73){\makebox(0,0)[cc]{$g$}}
\put(65,73){\vector(0,-1){10.5}}

\put(65,62){\circle*{1.5}}
\put(70,62){\makebox(0,0)[cc]{$x_1$}}

\put(64.5,62){\vector(2,-1){36.5}}

%\put(65,62){\circle*{1.5}}
\put(65,62){\line(-2,-1){4}}

\put(59,59){\line(-2,-1){4}}
\put(53,56){\line(-2,-1){4}}
\put(47,53){\line(-2,-1){4}}

\put(41,50){\line(-2,-1){4}}
%\put(35,47){\line(-2,-1){4}}
%\put(29,64){\line(-2,-1){4}}
%\put(47,73){\line(-2,-1){4}}

\put(35,47){\vector(-2,-1){7}}

\put(33,42){\makebox(0,0)[cc]{$x_2$}}
\put(27.5,43){\circle*{1.5}}

\put(17.5,30){\circle*{1.5}}
\put(22,30){\makebox(0,0)[cc]{$x_3$}}

\put(27,43){\line(-3,-4){4}}
\put(22,36){\vector(-3,-4){4}}

\put(17,30){\line(0,-1){4}}
\put(17,25){\line(0,-1){4}}
\put(17,20){\line(0,-1){4}}
\put(17,15){\line(0,-1){4}}
%\put(17,10){\line(0,-1){4}}
\put(17,10){\vector(0,-1){5}}

\put(27.5,43.5){\vector(3,-4){10}}
\put(38,30){\circle*{1.5}}
\put(43,30){\makebox(0,0)[cc]{$x_3$}}
\put(38,30){\line(0,-1){4}}
\put(38,25){\line(0,-1){4}}
\put(38,20){\line(0,-1){4}}
\put(38,15){\line(0,-1){4}}
%\put(17,10){\line(0,-1){4}}
\put(38,10){\vector(0,-1){5}}

\put(17,30){\vector(1,-2){12.5}}

\put(38,30){\vector(2,-3){16.5}}

\put(96,43){\makebox(0,0)[cc]{$x_2$}}
\put(101.5,43.5){\circle*{1.5}}
\put(113.5,31){\circle*{1.5}}
\put(108,30){\makebox(0,0)[cc]{$x_3$}}

\put(101.5,43.5){\vector(1,-1){12}}
\put(113.5,30.5){\vector(0,-1){25.5}}

\put(86,30){\makebox(0,0)[cc]{$x_3$}}
\put(91,30){\circle*{1.5}}

\put(101,43){\line(-3,-4){3}}
\put(97,38){\vector(-3,-4){5.5}}

\put(91,30){\vector(0,-1){25}}

\put(91,30){\line(-3,-4){3}}
\put(87,25){\line(-3,-4){3}}
\put(83,20){\line(-3,-4){3}}
\put(79,15){\line(-3,-4){3}}
\put(75,10){\vector(-3,-4){3.5}}

\put(113,30){\line(-1,-2){2}}
\put(110.5,25){\line(-1,-2){2}}
\put(108,20){\line(-1,-2){2}}
\put(105.5,15){\line(-1,-2){2}}
\put(103,10){\vector(-1,-2){2.5}}

\put(15,2){\framebox{0}}
\put(26,2){\framebox{1}}
\put(37,2){\framebox{0}}

\put(54,2){\framebox{1}}
\put(70,2){\framebox{0}}
\put(89,2){\framebox{0}}

\put(99,2){\framebox{1}}
\put(112,2){\framebox{1}}

\end{picture}

\caption{Decomposition tree for $g=x_1x_2\oplus x_1^0x_3.$}\label{f2}
\end{figure}

An {\it ordered decision diagram} of a function $f$ is obtained from the corresponding ODT by {\it reduction} of its nodes applying of the following two rules starting from the ODT and continuing until neither rule can be applied:

{\bf Reduction rules}
\begin{enumerate}
\item[$\bullet$] If two nodes are terminal and have the same label, or are non-terminal and have the same children, they are merged.
\item[$\bullet$] If an non-terminal node has identical children it is removed from the graph and its incoming edges are redirected to the child.
\end{enumerate}

When $k=2$ ODD is called {\it a binary  decision diagram} (BDD).  BDDs are extensively used in  the theory of {\it switching circuits} to represent and manipulate Boolean functions and to measure the complexity of binary terms.

The size of the ODD is determined  both by the function being represented and the chosen ordering of the variables.
It is of crucial importance to care about variable ordering when applying ODDs in practice. The problem of finding the best variable ordering is NP-complete  (see \cite{bra}). 

 {Figure} \ref{f3} shows the BDDs for the functions from Example \ref{ex2.1} obtained from their decomposition trees  under the natural ordering of their variables - $\l 1; 2; 3\r$. The construction of the ODT for $f$ under the natural ordering of the variables is left to the reader.

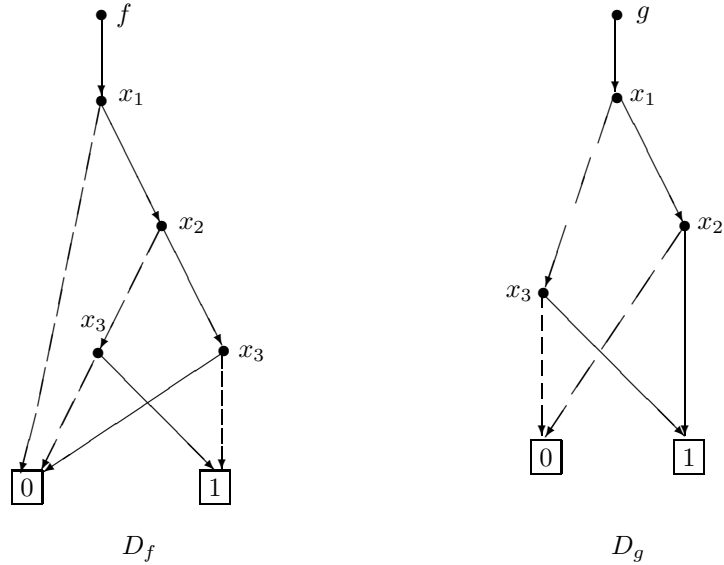
\begin{figure}[t]
\centering
\unitlength 1mm % = 2.845pt
\linethickness{0.4pt}
\begin{picture}(116.006,75.001)(0,0)
\put(25,73){\circle*{1.5}}
\put(28,73){\makebox(0,0)[cc]{$f$}}
\put(25,73){\vector(0,-1){10.75}}

\put(25,61.5){\circle*{1.5}}
\put(29,62){\makebox(0,0)[cc]{$x_1$}}

\put(24.5,62){\vector(1,-2){8.25}}

\put(33,45){\circle*{1.5}}
\put(37,45){\makebox(0,0)[cc]{$x_2$}}

\put(33,45){\vector(1,-2){8}}

\put(41.25,28.25){\circle*{1.5}}
\put(45,28){\makebox(0,0)[cc]{$x_3$}}

\put(33,45){\line(-1,-2){2}}
\put(30.5,40){\line(-1,-2){2}}
\put(28,35){\vector(-1,-2){3.25}}
\put(24.5,28){\circle*{1.5}}
\put(24,32){\makebox(0,0)[cc]{$x_3$}}

\put(24,27){\line(-1,-2){2}}
\put(21.5,22){\line(-1,-2){2}}
\put(19,16.5){\vector(-1,-2){2}}
\put(13,9){\framebox{0}}

\put(25,62){\line(-1,-5){1}}
\put(23.75,56){\line(-1,-5){1}}
\put(22.5,50){\line(-1,-5){1}}
\put(21.25,44){\line(-1,-5){1}}

\put(20,38){\line(-1,-5){1}}
\put(18.75,32){\line(-1,-5){1}}
\put(17.5,26){\line(-1,-5){1}}
%\put(16.25,12){\vector(-1,-5){3}}
\put(16.25,20){\vector(-1,-4){2}}

\put(24.5,28){\vector(1,-1){15.5}}
\put(41,28){\vector(-3,-2){23.75}}

\put(38,9){\framebox{1}}

\put(41,28){\line(0,-1){2}}
\put(41,25.5){\line(0,-1){2}}
\put(41,23){\line(0,-1){2}}
\put(41,20.5){\line(0,-1){2}}
\put(41,18){\line(0,-1){2}}
\put(41,15.5){\vector(0,-1){3}}

\put(30,2){\makebox(0,0)[cc]{$D_f$}}
%%%%%%%%%%%%%%%%%%%%%%%%%%%%%%%%%%%%%%%%%%%%%%%%%%%%%%%%%%%%%%%%%%%%
\put(93.5,73){\circle*{1.5}}
\put(97,73){\makebox(0,0)[cc]{$g$}}
\put(93.25,73){\vector(0,-1){10.25}}

\put(93.5,62){\circle*{1.5}}
\put(97,62){\makebox(0,0)[cc]{$x_1$}}

\put(94,62){\vector(1,-2){8.25}}

\put(102.5,45){\circle*{1.5}}
\put(106,45){\makebox(0,0)[cc]{$x_2$}}
\put(102.5,45){\vector(0,-1){28}}

\put(102.5,45){\line(-2,-3){3}}
\put(98.5,39){\line(-2,-3){3}}
\put(94.5,33){\line(-2,-3){3}}

\put(90.5,27){\line(-2,-3){3}}
%\put(93,17){\line(-1,-2){2}}
\put(86.75,21){\vector(-2,-3){2.75}}
%\put(88,7){\vector(-1,-2){2}}
\put(82,13){\framebox{0}}

\put(83.75,36){\circle*{1.5}}
\put(80.5,36){\makebox(0,0)[cc]{$x_3$}}
\put(83.5,36){\vector(1,-1){19}}

\put(83.5,36){\line(0,-1){2}}
\put(83.5,33){\line(0,-1){2}}
\put(83.5,30){\line(0,-1){2}}
\put(83.5,27){\line(0,-1){2}}
\put(83.5,24){\line(0,-1){2}}
\put(83.5,21){\vector(0,-1){4}}
%\put(83,8){\vector(0,-1){2}}
%\put(83,5){\vector(0,-1){2}}

\put(93,62){\line(-1,-3){2}}
\put(90,54){\line(-1,-3){2}}
\put(86.75,45){\vector(-1,-3){2.75}}

%\put(111,20){\vector(-3,-2){24}}

\put(101,13){\framebox{1}}

\put(95,2){\makebox(0,0)[cc]{$D_g$}}

\end{picture}

\caption{BDD for $f$ and $g$ under the natural ordering of variables.}\label{f3}
\end{figure}

The BDD of $f$ is more complex than the BDD of $g$. This reflects the fact that
$f$ has more separable pairs.
 Thus we have  $M=\{x_2,x_3\}\notin Sep(g)$,   $\{x_1\}\in Dis(M,g)$ and $\{x_1\}\in Sys(Dis(M,g))$. Additionally, the diagram of $g$ starts with $x_1$ - a variable which belongs to an $s$-system of $Dis(M,g)$. In this simple case we have $Sys(Dis(M,g))= Dis(M,g)=\{x_1\}$.

Figure \ref{f3} shows that when constructing the ODD of a function, it is better to start with the variables from an $s$-system of the family of  distributive sets of an inseparable set $M$ in this function. In \cite{ivo2} it is shown that the BDDs of functions have to be most simple when starting with variables from $Sys(Dis(M,f))$.
 Consequently, the inseparable sets with their distributive sets are important in   theoretical and applied computer science concerning the computational complexity of the functions.

Next, we  define and explore   complexity measures of the functions in $P_k^n$ which are directly connected with the computational complexity of functions.
We might think that the complexity of a function $f$ depends on the complexity of its ODDs. 

Let ${f\in P_k^n}$ and let $DD(f)$ be the set of the all ODDs for $f$ constructed under different variable orderings in $f$.
  
\begin{definition}\label{d3.2} 
Each path starting from the function node and finishing into a terminal node is called an {\it implementation} of the function $f$ under the given variable ordering.
 The set of the all implementations of $D_f$ we denote by $Imp(D_f)$ and 
\[Imp(f)=\bigcup_{D_f\in DD(f)}Imp(D_f).\]
 \end{definition}
 
Each implementation of the function ${f\in P_k^n}$, obtained from the diagram $D_f$ of $f$ by the non-terminal nodes 
$x_{i_1},\ldots,x_{i_r}$ and corresponding constants $c_{1},\ldots,c_{r},c\in Z_k$ with
$f(x_{i_1}=c_{1},\ldots,x_{i_r}=c_{r})=c,\quad r\leq ess(f)$, 
 can be represented as a pair $\mathbf{(i,c)}$ of two words (strings) over $\mathbf{n}=\{1,\ldots,n\}$ and $Z_k$ where $\mathbf{i}=i_1i_2\ldots i_r\in \mathbf{n}^*$ and $\mathbf{c}=c_1c_2\ldots c_rc\in Z_k^*$.
  
There is an obvious way to define a measure of complexity of a given ordered decision diagram $D_f$, namely as the number ${imp}(D_f)$ of all paths in  $D_f$ which starts from the function node and finish in a  terminal node of the diagram. 

The {\it implementation complexity} of a function ${f\in P_k^n}$ is defined as the number of all implementations of $f$, i.e.
$imp(f)=|Imp(f)|.$

We shall study also two other measures of computational complexity of functions as $sub(f)$ and $sep(f)$.

\begin{example}\label{ex3.1} Let us  consider again the functions $f$ and $g$ from Example \ref{ex2.1}, namely
$f=x_1x_2\oplus x_1x_3\quad and\quad g=x_1x_2\oplus x_1^0x_3.$

Then $(123, 1011)$ is an implementation of $f$ obtained by the diagram $D_f$  presented in
{Figure} \ref{f3}, following the path $\pi= (f; x_1: 1; x_2: 0; x_3: 1;terminal\ node: 1)$.

It is easy to see that there are six distinct BDDs for $f$ and five distinct BDDs for $g$. We shall calculate the implementations of $f$ and $g$ for the variable orderings  $\l 1; 2; 3\r$ (see Figure \ref{f3}) and $\l 2; 1; 3\r$, only. Thus for $f$ we have:
\vspace{.25cm}

\noindent
\begin{tabular}{|l|c|}
\hline
 ordering & implementations\\ \hline
   $\l 1; 2; 3\r$ &
$(1,00);\ (123,1000);\  (123,1011);\  (123,1101);\  (123, 1110)$ \\ \hline
$\l 2; 1; 3\r$& $(21,000);\  (213,0100);\ (213,0111);\ (21,100);  (213,1101);\ (213,1110)$ \\ \hline
\end{tabular}\vspace{.25cm}

For the function $g$ we obtain:\vspace{.25cm}

\noindent
\begin{tabular}{|l|c|}
\hline
 ordering &implementations\\ \hline
  $\l 1; 2; 3\r$ &
$(13,000);\ (13,011);\  (12,100);\  (12,111)$\\ \hline
$\l 2; 1; 3\r$& $(21,010);\  (213,0000);\ (213,0011);\ (213,1000);  (213,1011);\ (21,111)$\\ \hline
\end{tabular}\vspace{.25cm}

For the diagrams  in  
 {Figure} \ref{f3} we have 
${imp}(D_f)=5$ and ${imp}(D_g)=4$.

  Since $f$ is a symmetric function with respect to  $x_2$ and $x_3$ one can count  that $imp(f)=33$.
Note that the implementation  $(1,00)$ occurs in two distinct diagrams of $f$, namely under the orderings $\l 1; 2; 3\r$ and $\l 1; 3; 2\r$. Hence, it has to be counted one time and we obtain that $imp(f)$ is equal to $33$ instead of $34$.

For the function $g$, the diagrams under the  orderings $\l 1; 2; 3\r$ and $\l 1; 3; 2\r$ have the same implementations, i.e. the diagrams are identical (isomorphic). This fact is a consequence of   inseparability of the set $\{x_2,x_3\}$. Hence $g$  has five (instead of six for $f$) distinct ordered decision diagrams. Then,  one might calculate that $imp(g)=28$.

For the other two measures of complexity we obtain: 
$sub(f)=13$ because of 
$Sub(f)=\{0,1,x_1,x_2,x_3,x_2^0,x_3^0,x_2\oplus x_3,x_1x_2,$ $ x_1x_2^0,x_1x_3,x_1x_3^0, f\}$
and $sub(g)=11$ because of 
$Sub(g)=\{0,1,x_1,x_2,x_3,x_1^0,x_1 x_2,x_1^0x_3,x_1\oplus x_1^0x_3,x_1x_2\oplus x_1^0,g\}$.
Furthermore, $sep(f)=7$ because of 
 $M\in Sep(f)$ for all $M$, $\emptyset\neq M\subseteq \{x_1,x_2,x_3\}$ and  $sep(g)=6$ because of 
$Sep(g)=\{\{x_1\},\{x_2\},\{x_3\},\{x_1,x_2\},\{x_1,x_3\},\{x_1,x_2,x_3\}\}$.
\end{example}

\begin{lemma}\label{l17}
A variable $x_i$ is essential in $f\in P_k^n$ if and only if $x_i$ occurs as  a label of an non-terminal node in any ODD of $f$.
\end{lemma}
\begin{proof}
$"\Ra"$
Let us assume that $x_i$ does not occur as a label of any non-terminal node in an ordered decision diagram $D_f$ of $f$. Since all values of the function $f$ can be obtained by traversal walk-trough all paths in $D_f$ from function node to leaf nodes this will mean that $x_i$ will not affect the function value and hence $x_i$ is an inessential variable in $f$.

%%%%%
$"\La"$ Let $x_i\notin Ess(f)$ be an inessential variable in $f$. It is obvious  that for each subfunction $g$ of $f$ we have $x_i\notin Ess(g)$. Then we have $f(x_i=c)=f(x_i=d)$  for all $c,d\in Z_k$. Consequently, if there is  a non-terminal node labelled by $x_i$ in an ODT of $f$ then   its children have to be identical, which shows that this node has to be removed from the ODT, according to the reduction rules, given above.   
\end{proof}

An essential variable $x_i$ in a function $f$ is called {\it a strongly essential variable } in $f$ if there is a constant $c\in Z_k$ such that $Ess(f(x_i=c))=Ess(f)\setminus\{x_i\}$.

\begin{fact}\label{fc1}
If $ess(f)\geq 1$ then there is at least one strongly essential variable in $f$.
\end{fact}
This fact was  proven by O. Lupanov \cite{lup} in case of Boolean functions and by A. Salomaa  \cite{sal} for arbitrary functions.
 Later, Y. Breitbart \cite{bre} and  K. Chimev \cite{ch51}   proved that if $ess(f)\geq 2$ then there exist at least two strongly essential variables in $f$.
 We need Fact \ref{fc1} to prove the next important theorem.

\begin{theorem}\label{t3.2} A non-empty set $M$ of essential variables  is separable in $f$ if and only 
 if there exists  an implementation $\mathbf{(j,c)}$ of the form \[\mathbf{(j,c)}=(j_1j_2\ldots j_{r-m}j_{r-m+1}\ldots j_r, c_1c_2\ldots c_{r-m}c_{r-m+1}\ldots c_r c)\in Imp(f)\] where $M=\{x_{j_{r-m+1}},\ldots, x_{j_r}\}$ and $1\leq m\leq r\leq ess(f)$.
\end{theorem}
\begin{proof}
"$\La$" Let \[\mathbf{(j,c)}=(j_1\ldots j_{r-m}j_{r-m+1}\ldots j_r, c_1\ldots c_{r-m}c_{r-m+1}\ldots c_rc)\in Imp(f)\] be an implementation of $f$ and let $M=\{x_{j_{r-m+1}},\ldots, x_{j_r}\}$. Hence the all variables from $\{x_{j_{r-m+1}},\ldots,x_{j_r}\}$ are essential in the following subfunction of $f$ 
\[g=f(x_{j_1}=c_{1},\ldots,x_{j_{r-m}}=c_{{r-m}})\] which shows that $M\in Sep(f)$.

"$\Ra$"  Without loss of generality let us assume that  $M=\{x_1,\ldots,x_m\}$ is a non-empty separable set in $f$ and $n=ess(f)$. Then there are constants $d_{m+1},\ldots,d_n\in Z_k$ such that $M=Ess(h)$ where $h=f(x_{m+1}=d_{m+1},\ldots,x_{n}=d_{n})$. From Fact  \ref{fc1} it follows that there is a variable $x_{i_1}\in M$ and a constant $d_{1}\in Z_k$ such that $Ess(h_1)=M\setminus \{x_{i_1}\}$ where $h_1=h(x_{i_1}=d_{1})$. Consequently,  we might inductively obtain that there are variables $x_{i_r}\in M$ and constants $d_{r}\in Z_k$  for $r=2,\ldots,m$, such that  $Ess(h_r)=M\setminus \{x_{i_1},\ldots,x_{i_r}\}$ where $h_r=h_{r-1}(x_{i_r}=d_{r})$. Hence, the string $m+1m+2\ldots n$ has a substring $j_1\ldots j_s$ such that   $(j_1\ldots j_si_1\ldots i_m, d_{j_1}\ldots d_{j_s}d_{1}\ldots d_{m}d)$ is an implementation of $f$ with $M=\{x_{i_1},\ldots,x_{i_m}\}$ and $d=h_m$.
\end{proof}
\begin{corollary}\label{c3.2} For each variable $x_i\in Ess(f)$ there is an implementation $\mathbf{(j,c)}$ of $f$ whose last letter of $\mathbf{j}$ is $i$, i.e. $\mathbf{(j,c)}=(j_1\ldots j_{m-1}i, c_{j_1}\ldots c_{j_{m-1}}c_ic)\in Imp(f)$, $m\leq ess(f)$.
\end{corollary}
Note that there exists an ODD of a function whose non-terminal nodes are labelled by the variables from a given set, but this set might not be separable. For instance, the implementation $(231,0101)\in Imp(g)$ of the function $g$ from Example \ref{ex3.1} shows that the variables from the set  $M=\{x_2,x_3\}$  occur as labels of the starting two non-terminal nodes  in the BDD of $g$ under the ordering $\l 2; 3; 1 \r$, but $M\notin Sep(g)$.
%%%%%%%%%%%%%%%%%%%%%%%%%%%%%%%%%%%%%%%%%%%%%%%%%%%%%%%%%%%%%%%%%%%

 \begin{lemma}\label{l2.1}
If $ess(f)=n$, $g\preceq f$  with
$ess(g)= m<n$ then there exists a variable $x_t\in Ess(f)\setminus Ess(g)$ such that $Ess(g)\cup\{x_t\}\in Sep(f)$.
 \end{lemma}
\begin{proof}
Let $M=Ess(g)$. Then $M\in Sep(f)$ and from Theorem \ref{t3.2} it follows that there is an implementation $\mathbf{(j,c)}$ of the form $\mathbf{(j,c)}=(j_1j_2\ldots j_{r-m}j_{r-m+1}\ldots j_r,$ $c_1c_2\ldots c_{r-m}c_{r-m+1}\ldots c_r c)\in Imp(f)$ where $M=\{x_{j_{r-m+1}},\ldots, x_{j_r}\}$ and $1\leq m\leq r\leq ess(f)$. Since $m<n$ it follows that $r-m>0$ and Lemma \ref{l17} shows that there is $x_{j_i}\in Ess(h)$ where \[h=f(x_{j_1}=c_1,\ldots,x_{j_{i-1}}=c_{i-1},x_{j_{i+1}}=c_{i+1},\ldots x_{j_{r-m}}=c_{r-m}).\] It is easy to see that $Ess(h)=M\cup\{x_{j_i}\}$.
\end{proof}

Now, as an immediate consequence of the above lemma we obtain Theorem \ref{t2.4} which
was inductively proven by K. Chimev. 
\begin{theorem}\label{t2.4}  \cite{ch51}
If $ess(f)=n$, $g\preceq f$  with
$ess(g)= m\leq
n$ then there exist $n-m$ subfunctions $g_1,\ldots,g_{n-m}$ such that
\[g \prec g_1 \prec g_2\prec \ldots\prec g_{n-m}= f\]
and $ess(g_i)=m+i$ for $i=1,\ldots,n-m$.
\end{theorem}

%%%%%%%%%%%%%%%%%%%%%%%%%%%%%%%%%%%%%%

 The {\it depth}, (denoted by   $Depth(D_f)$) of an ordered decision diagram $D_f$ for a function $f$ is defined as  the number of the edges in a longest path from the function node in $D_f$ to a leaf of $D_f$.

Thus for the diagrams in Figure \ref{f3} we have $Depth(D_f)=4$ and $Depth(D_g)=3$.

Clearly, if $ess(f)=n$ then $Depth(D_f)\leq n+1$ for all ODDs of $f$.

\begin{theorem}\label{l3.2}
If $ess(f)=n\geq 1$ then there is an ordered decision diagram $D_f$ of $f$ with $Depth(D_f)=n+1$.
\end{theorem}
\begin{proof}
Let $Ess(f)=\{x_1,\ldots,x_n\}$, $n\geq 1$. Since $x_1$ is an essential variable it follows that $\{x_1\}\in Sep(f)$. Theorem \ref{t2.4} implies that there is an ordering $\l i_1; i_2; \ldots; i_{n-1}\r$  of the rest variables $x_{2}, \ldots, x_{{n}}$ such that for each $j$, $1\leq j\leq n-1$ we have $g_j\prec_J^{\mathbf{c}} f$ where $J=\{x_{i_1},\ldots,x_{i_j}\}$, $\mathbf{c}\in Z_k^{J}$ and $Ess(g_j)=\{x_1,x_{i_{j+1}},\ldots,x_{i_{n-1}}\}$. This shows that the all variables from  $J$ have to be labels of  non-terminal nodes in a path $\pi$ of the ordered decision diagram $D_f$ of $f$ under the  variable ordering $\l i_1; i_2; \ldots; i_{n-1}; 1\r$. Hence $\pi$ has to contain all essential variables in $f$ as labels at the non-terminal nodes of $\pi$. Hence $Depth(D_f)=n+1$.
\end{proof}

%\vspace{.05cm}

\begin{theorem}\label{t3.3}
Let $f\in P_k^n$ and $Ess(f)=\{x_1,\ldots,x_n\}$, $n\geq 1$. If $M\neq \emptyset$, $M\subset Ess(f)$ and $M\notin Sep(f)$ then there is a  decision diagram $D_f$ of $f$ with $Depth(D_f)<n+1$.
\end{theorem}
\begin{proof}
Without loss of generality, let us assume that $M=\{x_1,\ldots,x_m\}$, $m<n$. Since $M$ is  inseparable  in $f$, the family $Dis(M,f)$ of the all  distributive sets of $M$ is non-empty. According to Theorem \ref{t2.2} there is a non-empty $s$-system $\beta=\{x_{i_1},\ldots,x_{i_t}\}$ of $Dis(M,f)$. Since $f(x_{i_1}=c_1)\neq f(x_{i_1}=c_2)$ for some $c_1,c_2\in Z_k$ it follows that there exists an ODD $D_f$ for $f$ under a variable ordering  with $x_{i_1}$ as the label of the first non-terminal node of $D_f$. According to Corollary \ref{c2.1} for all $c\in Z_k$ there is a variable $x_j\in M$ which is inessential in $f(x_{i_1}=c)$. Hence, each path of $D_f$ does not contain at least one variable from $M$ among its labels of non-terminal nodes. Hence $Depth(D_f)<n+1$.
\end{proof}

%%%%%%%%%%%%%%%%%%%%%%%
\section{Equivalence Relations and Transformation Groups in $P_k^n$}\label{sec5}

Many of the problems in   applications of the $k$-valued functions are compounded because of the large number of the functions, namely $k^{k^n}$. Techniques which involve enumeration of functions can only be used if $k$ and $n$ are trivially small. A common way for extending the scope of such enumerative methods is to classify the functions into equivalence classes under some natural equivalence relation.

In this section   we  define equivalence relations in $P_k^n$ which classify functions with respect to number of their implementations, subfunctions and separable  sets. We are intended to determine several numerical invariants of the transformation groups generated by these relations. The second goal is to compare these groups with so called classical subgroups of the  Restricted Affine Group(RAG) \cite{lech} which have a variety of applications such as coding theory,  switching theory, multiple output combinational logic, sequential machines and other areas of theoretical and applied computer sciences.

Let us denote by $S_A$ the symmetric group of all permutations of a given no-empty set $A$.  $S_m$  denotes the symmetric group $S_{\{1,\ldots,m\}}$ for a natural number $m$, $m\geq 1$.

Let us define the following three equivalence relations: $\simeq_{imp}$, $\simeq_{sub}$ and $\simeq_{sep}$.
\begin{definition}\label{d5.1} Let  $f,g\in P_k^n$ be two functions.
\begin{enumerate}
\item[(i)] If $ess(f)=ess(g)\leq 1$ then $f\simeq_{imp} g$;
\item[(ii)] Let $ess(f)=n>1$. We say that $f$ is $imp$-equivalent to $g$ (written $f\simeq_{imp} g$) if  there are $\pi\in S_n$ and  $\sigma_i\in S_{Z_k}$ such that $f(x_i=j)\simeq_{imp} g(x_{\pi(i)}=\sigma_i(j))\quad\mbox{for all}\quad   i=1,\ldots,n\quad\mbox{and}\quad j\in Z_k.$
\end{enumerate}
\end{definition}
Hence two functions are ${imp}$-equivalent if they produce same number of  implementations, i.e.
    $imp(f)=imp(g)$ and there are $\pi\in S_n$, and  $\sigma$, $\sigma_i\in S_{Z_k}$ such that 
$(i_1\ldots i_m,c_{1},\ldots,c_{m}c)\in Imp(f)\iff $ $ (\pi(i_1)\ldots \pi(i_m), \sigma_1(c_{1})\ldots\sigma_m(c_{m})\sigma(c))\in Imp(g).$

Table \ref{tb1} shows the classification of Boolean functions of two  variables into four classes, called {\it imp-classes} under the equivalence relation $\simeq_{imp}$. 
The second column shows the number of implementations of the functions from the $imp$-classes given at the first column. The third column presents the number of functions per each $imp$-class.
% \vspace{-.25cm}
 
\begin{table}[h]
\caption{$Imp$-classes in  $P_2^2$.} \label{tb1}
%\vspace{.5cm}
\centering
\begin{tabular}{l|l|l}%\hline
\hline\hline
%Class & Number  \\ 
%&  of implementations\\ %\hline\hline
$[\ 0,\ 1\ ]$ & \ 1& 2\\ \hline
$[\ x_1,\ x_2,\ x_1^0,\ x_2^0\ ]$&\   2& 4\\  \hline
$[\ x_1x_2,\ x_1x_2^0,\ x_1^0x_2,\ x_1^0x_2^0,\ x_1\oplus x_1x_2,$ &&\\
$ x_2^0\oplus x_1x_2,\ x_1^0\oplus x_1x_2,\ x_1^0\oplus x_1x_2^0\ ]$& \ 6 & 8\\  \hline
$[\ x_1\oplus x_2,\ x_1\oplus x_2^0\ ]$&\   8& 2\\ \hline\hline
\end{tabular}
\end{table}
%\vspace{.15cm}

\begin{definition}\label{d5.2} Let  $f,g\in P_k^n$ be two functions.
\begin{enumerate}
\item[(i)] If $ess(f)=ess(g)=0$ then $f\simeq_{sub} g$;
\item[(ii)] If  $ess(f)=ess(g)=1$ then $f\simeq_{sub} g\iff range(f)=range(g)$;
\item[(iii)] Let $ess(f)=n>1$. We say that $f$ is $sub$-equivalent to $g$ (written $f\simeq_{sub} g$) if   $sub_m(f)=sub_m(g)$ 
for all $m=0,1,\ldots, n$.
\end{enumerate}
\end{definition}
It is easy to see that the equivalence relation $\simeq_{sub}$ partitions the algebra of Boolean functions of two variables in the same equivalence classes (called {\it the sub-classes}) as  the relation $\simeq_{imp}$ (see Table \ref{tb1}). 
\begin{definition}\label{d5.3} Let  $f,g\in P_k^n$ be two functions.
\begin{enumerate}
\item[(i)] If $ess(f)=ess(g)\leq 1$ then $f\simeq_{sep} g$;
\item[(ii)]Let $ess(f)=n>1$. We say that $f$ is $sep$-equivalent to $g$ (written $f\simeq_{sep} g$) if   $sep_m(f)=sep_m(g)$
for all $m=1,\ldots, n$.
\end{enumerate}
\end{definition}
The equivalence classes under $\simeq_{sep}$ are called {\it sep-classes}.

 Since $P_k^n$ is a finite algebra of $k$-valued functions each equivalence relation $\simeq$ on $P_k^n$  makes a partition of the algebra in the set of disjoint equivalence classes $Cl(\simeq)=\{P_1^\simeq,\ldots,P_r^\simeq\}$. Then, in the set of all equivalence relations a partial order is defined as follows: $\simeq_1\ \leq\ \simeq_2$ if for each $P\in Cl(\simeq_1)$ there is a $Q\in Cl(\simeq_2)$ such that  $P\subseteq Q$. Thus $\simeq_1\ \leq\ \simeq_2$ if and only if   $f\simeq_1 g\ \Ra f\simeq_2 g$, for  $f,g\in P_k^n$.

 \begin{theorem}\label{t5.1}
 ~~~
\begin{enumerate}
\item[(i)] $ \simeq_{imp}\  \leq\ \simeq_{sep}$; \quad (iii) $\simeq_{imp}\  \not\leq\ \simeq_{sub}$;
\item[(ii)] $ \simeq_{sub}\  \leq\ \simeq_{sep}$; \quad (iv) $ \simeq_{sub}\  \not\leq\ \simeq_{imp}$.
\end{enumerate}

 \end{theorem}

\begin{proof} 
(i)\  
Let $f,g\in P_k^n$ be two ${imp}$-equivalent functions, i.e. $f\simeq_{imp} g$. 
We shall proceed by induction on the number $n=ess(f)$ of essential variables in $f$ and $g$.

Clearly, if $n\leq 1$ then $f\simeq_{sub} g$, which is our inductive basis.
Let us assume that   $f\simeq_{imp} g$ implies $f\simeq_{sep} g$ if $n< r$ for some natural number $r$, $r\geq 2$.

Let $f$ and $g$ be two functions with $f\simeq_{imp}g$ and $ess(f)=ess(g)=r$. Then there are $\pi\in S_r$ and $\sigma_i\in S_{Z_k}$ for $i=1,\ldots,r$ such that $f(x_i=j)\simeq_{imp} g(x_{\pi(i)}=\sigma_i(j))$. Let $M$, $\emptyset\neq M\in Sep(f)$ be a separable set of essential variables in $f$ with $|M|=m$, $1\leq m\leq r$. Theorem \ref{t3.2} implies that there is an implementation 
\[\mathbf{(j,c)}=(j_1\ldots j_{r-m} i_1\ldots i_m, c_{j_1}\ldots c_{j_{r-m}}c_{i_1}\ldots c_{i_m}c)\]
of $f$ obtained  under an ODD  whose variable ordering  finishes with the variables from $M$, i.e. $M=\{x_{i_1}\ldots,x_{i_m}\}$. Then $f(x_{j_1}=c_{j_1})\simeq_{imp} g(x_{\pi(j_1)}=\sigma_{j_1}(c_{j_1}))$ implies that 
\[({\pi({j_1})}\ldots {\pi({j_{r-m}})}{\pi({i_1})}\ldots {\pi({i_m})},\sigma_{j_1}(c_{j_1})\ldots \sigma_{j_{r-m}}(c_{j_{r-m}})\sigma_{i_1}(c_{i_1})\ldots \sigma_{i_m}(c_{i_m})\sigma(c))\]
 is an implementation of $g$, for some $\sigma\in S_{Z_k}$. Again, from Theorem \ref{t3.2} it follows that 
$\pi(M)=\{x_{\pi(i_1)},\ldots,x_{\pi(i_m)}\}\in Sep(g).$  Since $\pi$ is a permutation of $S_r$ it follows that  $sep_m(f)=sep_m(g)$ for $m=1,\ldots,r$ and hence $\simeq_{imp}\ \leq\ \simeq_{sep}.$
 
(ii) \
Definition \ref{d1.3} shows that $M\in Sep(f)$ if and only if there is a subfunction $g\in Sub(f)$ with $g\prec_Q^{\mathbf{c}}f$ where $Q=Ess(f)\setminus M$ and $\mathbf{c}\in Z_k^{n-|M|}$. Hence
\[\forall f,g\in P_k^n,\quad Sub(f)=Sub(g)\implies Sep(f)=Sep(g),\]
which implies that 
$sub_m(f)=sub_m(g)\implies sep_m(f)=sep_m(g)$ and
$\simeq_{sub}\leq \simeq_{sep}$.

(iii)\
Let us consider the functions 
\[f=x_1^0x_2x_3\oplus x_1x_2^0x_3^0\ (mod\ 2)\quad\mbox{and}\quad g=x_2x_3\oplus x_1x_2^0x_3\oplus x_1x_2x_3^0\ (mod\ 2).\]

The set of the all simple subfunctions in $f$ is: $\{x_1x_2^0, 
x_1^0x_2, x_1x_3^0, x_1^0x_3, x_2x_3, x_2^0x_3^0\}$ and in $g$ is: $\{x_1x_2, 
x_1x_3, x_2x_3, x_2\oplus x_1x_2^0, x_3\oplus x_1x_3^0, x_2^0x_3^0\oplus 1\}$.

Hence $f$ and $g$ have six simple subfunctions, which depends essentially on two variables.  Table \ref{tb1} shows that all these subfunctions belong to same $imp$-class and the number of their implementations is $6$. Thus we might  calculate that $imp(f)=imp(g)=36$ and $f\simeq_{imp}g$.

The set of the all subfunctions with one essential variable in the function $f$ is: $\{ x_1, x_2, x_3, x_1^0, x_2^0, x_3^0\}$ and in $g$ is: $\{x_1, x_2, x_3\}$.

Then we have $sub_0(f)=sub_0(g)=2$, $sub_1(f)=6$, $sub_1(g)=3$ and $sub_2(f)=sub_2(g)=6$ and hence $f\not\simeq_{sub}g$.
It is clear that $sub(f)=15$, $sub(g)=12$  and $\simeq_{imp}\  \not\leq\ \simeq_{sub}$.

(iv)\
Let us consider the functions 
\[f=x_1x_2^0x_3^0\oplus x_1\ (mod\ 2)\quad\mbox{and}\quad g= x_1x_2x_3\ (mod\ 2).\]

The simple subfunctions in $f$ and $g$ are:\\ 
\begin{tabular}{lll}
$f(x_1=0)=0,$&$f(x_3=0)=x_1x_2^0\oplus x_1,$&  $g(x_2=0)=0,$ \\
$f(x_1=1)=x_2^0x_3^0\oplus 1,$& $f(x_3=1)=x_1,$ & $g(x_2=1)=x_1x_3,$\\
$f(x_2=0)=x_1x_3^0\oplus x_1,$ & $g(x_1=0)=0,$&$g(x_3=0)=0,$\\
 $f(x_2=1)=x_1,$ & $g(x_1=1)=x_2x_3,$&$g(x_3=1)=x_1x_2$.\\
  \end{tabular} 

Now, using Table \ref{tb1}, one can easily calculate that  $imp(f)=23$ and $imp(g)=21$, and hence $ f  \not\simeq_{imp} g$.
On the other side we have 
$Sub(f)=\{0, 1, x_1, x_2, x_3, x_2^0x_3^0\oplus 1, x_1x_3^0\oplus x_1, x_1x_2^0\oplus x_1, f\}$
and
$Sub(g)=\{0, 1, x_1, x_2, x_3, x_2x_3, x_1x_3, x_1x_2, g\}$
which show that 
$sub_m(f)=sub_m(g)\quad\mbox{for}\quad m=0,1,2,3$ and $f\simeq_{sub} g$.
Hence $ \simeq_{sub}\  \not\leq\ \simeq_{imp}$.
\end{proof}

A {\it  transformation} $\psi:P_k^n\longrightarrow P_k^n$  can be viewed  as  an $n$-tuple of functions
\[\psi=(g_1,\ldots,g_n),\quad g_i\in P_k^n,\quad i=1,\ldots,n\]
acting on any function $f=f(x_1,\ldots,x_n)\in P_k^n$ as follows
$\psi(f)=f(g_1,\ldots,g_n)$.
Then the  composition of two transformations $\psi$ and $\phi=(h_1,\ldots,h_n)$ is defined as follows
\[\psi\phi=(h_1(g_1,\ldots,g_n),\ldots,h_n(g_1,\ldots,g_n)).\]

Thus the set of all transformations of  $P_k^n$ is the {\it universal monoid $\Omega_k^n$} with unity - the identical transformation. When taking only invertible transformations we obtain the {\it universal group} $C_k^n$ isomorphic to the symmetric group $S_{Z_k^n}$.
Throughout this paper  we shall consider   invertible transformation, only. The groups consisting of invertible transformations of $P_k^n$ are called {\it transformation groups}.

Let $\simeq$ be an equivalence relation in $P_k^n$.
A mapping $\varphi:P_k^n\longrightarrow P_k^n$ is called {\it a transformation, preserving $\simeq$} if $f\simeq \varphi(f)$ for all $f\in P_k^n$.
Taking only invertible transformations which preserve $\simeq$, we get the group $G$ of all transformations preserving $\simeq$, whose {\it  orbits} (also called {\it $G$-types})  are the equivalence classes $P_1,\ldots,P_r$ under $\simeq$.
The number of orbits of a group $G$ of transformations in finite algebras of functions is denoted by $t(G)$.

Next, we relate groups to combinatorial problems trough the following obvious, but important definition:
\vspace{.1cm}

\begin{definition}\label{d5.4}~Let $G$ be a transformation group acting on the algebra of functions $P_k^n$and suppose that $f,g\in P_k^n$. We say that $f$ is $G$-equivalent to $g$ (written $f\simeq_G g$) if there exists $\psi\in G$ so that $g=\psi(f)$.
\end{definition}
Clearly, the relation $\simeq_G $ is an equivalence relation. We summarize and extend the results for the "classical" transformation groups, following \cite{har2,lech,str3}, where these notions are used to study classification and enumeration in the algebra of boolean functions. Such groups are induced  under  the following notions of equivalence: complementation and/or permutation of the variables; any linear or affine function of the variables.
Since we want to classify functions from $P_k^n$ into equivalence classes, three natural problems occur.
%\begin{problem}\label{prob5.1}~~~~~

\begin{enumerate}
 \item[$\bullet$]  We ask for the number $t(G)$ of such equivalence classes. This problem will be partially discussed for the family of ``natural'' equivalence relations in the algebra of boolean functions. 
\item[$\bullet$] We ask for the cardinalities of the equivalence classes. This problem is important in applications as functioning the switching gates, circuits etc. For boolean functions of 3 and 4 variables we shall solve these two problems, also concerning  $imp$-, $sub$- and $sep$-classes.
 
\item[$\bullet$] We want to give a method which will decide the class to which an arbitrary function belongs. In some particular cases this problem will be discussed below. We also develop a class of algorithms for counting the complexities $imp$, $sub$ and $sep$ for each boolean function which allow us to classify the algebras $P_2^n$ for $n=2,3,4$ with respect to these complexities as group invariants. 
\end{enumerate}
These problems are very hard and for  $n\geq 5$ they are practically unsolvable.

%%%%%%%%%%%%%%%%%%%%%%%%%%%%%%%%%%%%%%%%%%%%%%%%%%%%%%%%%%%%

We use the denotation $\leq$ also, for order relation ``subgroup''. More precisely, $H\leq G$ if there is a subgroup $G'$ of $G$ which is isomorphic to $H$.

Let us denote by $IM_k^n$, $SB_k^n$ and $SP_k^n$ the transformation groups induced by the equivalence relations $\simeq_{imp}$, $\simeq_{sub}$ and $\simeq_{sep}$, respectively. 

Now, as a direct consequence of Theorem \ref{t5.1} we obtain the following proposition.

\begin{proposition}\label{c5.1}~ 

 \begin{enumerate}\item[(i)] $IM_k^n\leq SP_k^n$; \quad (iii) $IM_k^n\not\leq SB_k^n$; 
\item[(ii)] $SB_k^n\leq SP_k^n$; \quad (iv) $SB_k^n\not\leq IM_k^n$.
\end{enumerate}
\end{proposition}

We  deal with "natural" equivalence relations which involve  variables in some functions. Such relations induce permutations on the domain $Z_k^n$ of the functions. These mappings form a transformation group whose number of equivalence classes  can be determined.

The restricted affine group (RAG) is defined as a subgroup of the symmetric group on the direct sum of the vector space $Z_k^n$ of arguments of functions and the vector space $Z_k$ of their outputs. The group RAG permutes the direct sum  $Z_k^n+Z_k$ under restrictions which preserve single-valuedness of all functions from $P_k^n$. The equivalence relation induced by RAG is called {\it prototype equivalence relation}.

In the model of RAG an affine transformation operates on the domain or space of inputs $\mathbf{x}=(x_1,\ldots,x_n)$ to produce the output $\mathbf{y}=\mathbf{xA}\oplus \mathbf{c}$, which might be used as an input in  a function $g$. Its output $g(\mathbf{y})$ together with the function variables $x_1,\ldots,x_n$ are linearly combined by a range transformation which defines the image $f(\mathbf{x})$ as follows:
\begin{equation}\label{eq2}
f(\mathbf{x})=g(\mathbf{y})\oplus a_1x_1\oplus\ldots\oplus a_nx_n\oplus d=g(\mathbf{xA}\oplus \mathbf{c})\oplus \mathbf{a^tx}\oplus d
\end{equation}
where $d$ and $a_i$ for $i=1,\ldots,n$ are constants from $Z_k$.

Such a  transformation belongs to RAG if $\mathbf{A}$ is a non-singular matrix. The name RAG was given to this group by R. Lechner  in 1963 (see \cite{lech1}) and it was studied by Ninomiya (see \cite{nin}) who gave the name "prototype equivalence" to the relation it induces on the function space $P_k^n$. 

We want to extract basic facts about some of the subgroups of RAG which are "neighbourhoods" or "relatives" of our transformation groups - $IM_k^n$, $SB_k^n$ and $SP_k^n$.

First, we consider a group which is called $CA_k^n$ (complement arguments) and each transformation $\mb{j}\in CA_k^n$ is determined by an $n$-tuple from $Z_k^n$, i.e.
$CA_k^n=\{(j_1,\ldots,j_n)\in Z_k^n\}.$
Intuitively, $CA_k^n$ will complement some of the variables of a function. If $\mb{j}=(j_1,\ldots,j_n)$ is in $CA_k^n$, define
$\mb{j}(x_1,\ldots,x_n)=(x_1\oplus j_1,\ldots,x_n\oplus j_n).$
The group operation is sum mod $k$ and written $\oplus$. For example if $n=k=3$ and $\mb{j}=(2,1,0)$ then $\mb{j}(x_1,x_2,x_3)=(x_1\oplus 2,x_2\oplus 1,x_3)$ and $\mb{j}$ induces a permutation on $Z_3^3=\{0,1,2\}^3$. Then the following sequence of images: $\mb{j}: 000\ra 210\ra 120\ra 000$ 
determines the cycle $(0,21,15)$ and if we agree to regard each triple from $Z_3^3$ as a ternary number, then the permutation induced by $\mb{j}$ can be written in cyclic notation as $(0,21,15)(1,22,16)(2,23,17)(3,24,9)(4,25,10)(5,26,11)$ $(6,18,12)(7,19,13)(8,20,14).$
In \cite{har2} M. Harrison showed that  the boolean functions of two variables are grouped into seven classes under the group $CA_2^2$.

Another classification occurs when permuting arguments. If $\pi\in S_n$ then $\pi$ acts on variables by:
$\pi(x_1,\ldots,x_n)=(x_{\pi(1)},\ldots,x_{\pi(n)}).$
Each permutation induces a map on the domain $Z_k^n$. For instance the permutation $\pi=(1,2)$ induces a permutation on $\{0,1,2\}^3$ when considering the algebra $P_3^3$. Then we have $\pi: 010\ra 100 \ra 010$ and in cyclic notation it can be written as
\[(3,9)(4,10)(5,11)(6,18)(7,19)(8,20)(15,21)(16,22)(17,23).\]
$S_k^n$ denotes the transformation group induced by permuting of variables. It is clear that $S_k^n$ is isomorphic to $S_n$.

If we allow both complementations and permutations of the variables, then a transformation group, called $G_k^n$, is induced. The group action on variables is represented by 
$((j_1,\ldots,j_n),\pi)(x_1,\ldots,x_n)=(x_{\pi(1)}\oplus j_1,\ldots, x_{\pi(n)}\oplus j_n)$
where $j_m\in Z_k$ for $1\leq m\leq n$ and $\pi\in S_n$. The group $G_2^n$ is especially important in switching theory and other areas of discrete mathematics, since it is the symmetry group of the $n$-cube. The classification of the boolean functions under $G_2^2$ into six classes is shown in \cite{har2}. 

Let us allow a function to be equivalent to its complements as well as using equivalence under $G_k^n$. Then the transformation group  which is induced by  this equivalence relation is called the {\it genera} of $G_k^n$ and it is denoted by $GE_k^n$. Thus the equivalence relation $\simeq_{gen}$ which induces genera of $G_k^n$ is defined as follows $f\simeq_{gen}g\iff f\simeq_{G_k^n}g$ or $f=g\oplus j$ for some $j\in Z_k$. Then there exist only four equivalence classes in $P_2^2$, induced by  $GE_2^2$. 
 These classes are the same as the classes induced by the group $IM_2^2$ in the algebra $P_2^2$ (see \cite{har2} and Table \ref{tb1}, given above).

Next important classification is generated by equivalence relations which allow adding linear or affine functions of variables. In order to preserve the group property we shall consider invertible linear transformations and assume that $k$ is a prime number such that $LG_k^n$ the general linear group on an $n$-dimensional vector space is over the field $Z_k$. The transformation groups $LG_2^n$ and $A_2^n$ of linear and affine transformations in the algebra of  boolean functions are included in the lattice of the subgroups of RAG. We  extend this view to the functions from $P_k^n$. The algebra of boolean functions in the simplest case of two variables is classified in eight classes under $LG_2^2$ and in five classes under $A_2^2$.  Table \ref{tb1_1} presents both equivalence classes of boolean functions from $P_2^2$ under the transformation group $RAG$.
\begin{table}%[h]
\caption{Classes in  $P_2^2$ under $RAG$.} \label{tb1_1}
%\vspace{.5cm}
\centering
\begin{tabular}{l}%\hline
\hline\hline
%Class & Number  \\ 
%&  of implementations\\ %\hline\hline
$[\ 0,\ 
   1,\ 
 x_1,\ x_2,\ x_1^0,\ x_2^0,\  x_1\oplus x_2,\ x_1\oplus x_2^0\ ]$\\ \hline
$[\ x_1x_2,\ x_1x_2^0,\ x_1^0x_2,\ x_1^0x_2^0,\ 
  x_1\oplus x_1x_2,\ x_2^0\oplus x_1x_2,\
 x_1^0\oplus x_1x_2,\ x_1^0\oplus x_1x_2^0\ ]$\\ \hline\hline
\end{tabular}
\end{table}

The subgroups of RAG defined above are 
determined  by equivalence relations as it is shown in Table \ref{tb2}, where  $\mb{P}$ denotes a permutation matrix, $\mb{I}$ is the identity matrix, $\mb{b\mbox{ and }c}$ are vectors from $Z_k^n$ and $d\in Z_k$.
%\vspace{.15cm}

\begin{table}%[h]
\caption{Subgroups of RAG}\label{tb2}
%\vspace{.5cm}
\begin{tabular}{||l|l|l||}\hline\hline
Subgroup& Equivalence relations& Determination\\ \hline
RAG & Prototype equivalence& $\mb{A}$-non-singular\\ 
$GE_k^n$ & genus & $\mb{A}=\mb{P}$, $\mb{a}=\mb{0}$\\ 
 $CF_k^n$ & complement function & $\mb{A}=\mb{I}$, $\mb{a}=\mb{0}$,     $\mb{c}=\mb{0}$\\ 
  $A_k^n$ &affine transformation  &  $\mb{a}=\mb{0}$, $d=0$\\ 
  $G_k^n$ & permute \& complement & \\
  & variables  (symmetry types) & $\mb{A}=\mb{P}$, $\mb{a}=\mb{0}$, $d=0$\\
  $LF_k^n$ & add linear function & $\mb{A=I}$, $\mb{c=0}$, $d=0$\\
$CA_k^n$ & complement arguments & $\mb{A}=\mb{I}$, $\mb{a}=\mb{0}$, $d=0$\\ 
 $LG_k^n$ & linear transformation & $\mb{c}=\mb{0}$, $\mb{a}=\mb{0}$, $d=0$\\
  $S_k^n$ & permute variables & $\mb{A}=\mb{P}$, $\mb{c}=\mb{0}$, $\mb{a}=\mb{0}$, $d=0$\\ \hline\hline
\end{tabular}
\end{table}
It is naturally to ask which subgroups of RAG are subgroups of the groups  $IM_k^n$ or $SB_k^n$. The answer of this question is our next goal. 
\begin{example}\label{ex5.1}
Let 
$f=x_1x_2^0x_3\oplus x_1^0\quad\mbox{and}\quad g=x_1x_2^0x_3\oplus x_1x_2$ be two boolean functions.
Then 
\[sub_1(f)=sub_1(g)=3,\ sub_2(f)=sub(g)=3\quad\mbox{and}\quad sub_3(f)=sub_3(g)=1.\]
Hence  $f\simeq_{sub} g$.
In a similar way, it can be shown that $f\simeq_{imp} g$.
 The details are left to the reader. 

On the other side, one can prove  that there is no  transformation $\varphi\in RAG$ such that $\varphi(x_1^0)=x_1x_2$ (see Table \ref{tb1_1}) and hence there is  no affine transformation $\varphi\in RAG$  for which $g=\varphi(f)$.

Consequently, each group among  $IM_k^n$, $SB_k^n$ and $SP_k^n$ can not be a subgroup of $RAG$.
\end{example}
%%%%%%%%%%%%%%%%%%%%%%%%%%%%%%%%%%
 Table \ref{tb2} allows us to establish  the following fact.
\begin{fact}\label{fc2}
If $f$ and $g$ satisfy  (\ref{eq2}) with $\mb{A\notin \{0,P,I\}}$ or $\mb{a\neq 0}$ then $f\not\simeq_{imp} g$, $f\not\simeq_{sub} g$ and $f\not\simeq_{sep} g$.
\end{fact}
%%%%%%%%%%%%%%%%%%%%%%%%%%%%%%
\begin{proposition}\label{p5.1}~~
\begin{enumerate}
\item[(i)] $LG_k^n\not\leq SP_k^n$;\quad (ii) $LF_k^n\not\leq SP_k^n$;
\item[(iii)] $IM_k^n\not\leq RAG$;\quad (iv) $SB_k^n\not\leq RAG$.
\end{enumerate}
\end{proposition}
\begin{proof}
Immediate from Fact \ref{fc2} and Example \ref{ex5.1}.
\end{proof}
Let $\sigma:Z_k\longrightarrow Z_k$ be a mapping and $\psi_\sigma:P_k^n\longrightarrow P_k^n$ be a transformation of $P_k^n$ determined by $\sigma$ as follows
$\psi_\sigma(f)(\mathbf{a})=\sigma(f(\mathbf{a}))$
for all $\mathbf{a}=(a_1,\ldots,a_n)\in Z_k^n$.

\begin{theorem}\label{t5.2}  $\psi_\sigma\in IM_k^n$ and $\psi_\sigma\in SB_k^n$ if and only if $\sigma$ is a permutation of $Z_k$, $k>2$.
\end{theorem}
\begin{proof} "$\La$" Let $\sigma\in S_{Z_k}$ be a permutation of $Z_k$ and let $f$ be an arbitrary function with $ess(f)=n\geq 0$. We shall proceed by induction on $n$, the number of essential variables in $f$.

If $n=0$ then clearly $\psi_\sigma(f)$ is a constant and hence $f\simeq_{imp} \psi_\sigma(f)$ and $f\simeq_{sub} \psi_\sigma(f)$.

Assume that if $n<p$ then $f\simeq_{imp} \psi_\sigma(f)$ and $f\simeq_{sub} \psi_\sigma(f)$  for some natural number $p, p>0$. Hence $f(x_i=j)\simeq_{imp}\psi_\sigma(f(x_i=j))$ and $sub_m(f(x_i=j))=sub_m(\psi_\sigma(f(x_i=j)))$ for all $i\in\{1,\ldots,n\}$, $m\in\{1,\ldots,n-1\}$  and $j\in Z_k$.

Let $n=p$. Let  $x_i\in\{x_1,\ldots,x_n\}=Ess(f)$  and $j\in Z_k$, and let us set $g=f(x_i=j)$. Then $\psi_\sigma(g)=\psi_\sigma(f(x_i=j)$ and $ess(g)=n-1<p$. Hence our inductive assumption implies
$g\simeq_{imp}\psi_\sigma(g)$ and $g\simeq_{sub}\psi_\sigma(g)$. Consequently, we have
\[f(x_i=j)\simeq_{imp}\psi_\sigma(f(x_i=j))
\quad\mbox{and}\quad sub_m(f(x_i=j))=sub_m(\psi_\sigma(f(x_i=j)))\]
for all  $x_i\in\{x_1,\ldots,x_n\}$ and $j\in Z_k$, which shows that $f\simeq_{imp}\psi_\sigma(f)$ and $f\simeq_{sub}\psi_\sigma(f)$.

"$\Ra$" Let us assume that $\sigma$ is not a permutation of $Z_k$. Hence there exist two constants $a_1$ and $a_2$ from $Z_k$ such that $a_1\neq a_2$ and $\sigma(a_1)=\sigma(a_2)$. Let us fix the vector $\mathbf{b}=(b_1,\ldots,b_n)\in Z_k^n$. Then we define the following function from $P_k^n$:
\[
f(x_1,\ldots,x_n)=\left\{\begin{array}{ccc}
            a_1 \  &\  if \  &\  x_i=b_i\ for\ i=1,\dots,n \\
            a_2 &   & otherwise.
           \end{array}
           \right.
\]
Clearly, $Ess(f)=X_n$. On the other hand the range of $f$ is $range(f)=\{a_1,a_2\}$ and $\sigma(range(f))=\{\sigma(a_1)\}$, which implies that
$\psi_\sigma(f)(c_1,\ldots,c_n)=\sigma(a_1)$
for all $(c_1,\ldots,c_n)\in Z_k^n$. Hence $\psi_\sigma(f)$ is the constant $\sigma(a_1)\in Z_k$ and  $Ess(\psi_\sigma(f))=\emptyset$.
Thus we have $f\not\simeq_{imp} \psi_\sigma(f)$ and $f\not\simeq_{sub} \psi_\sigma(f)$. 
\end{proof}

\begin{theorem}\label{t5.3}
Let $\pi\in S_n$ and $\sigma_i\in S_{Z_k}$  for $i=1,\ldots,n$. Then 
$f(x_1,\ldots,x_n)\simeq_{imp} f(\sigma_1(x_{\pi(1)}),\ldots,\sigma_n(x_{\pi(n)}))$
and
$f(x_1,\ldots,x_n)\simeq_{sub} f(\sigma_1(x_{\pi(1)}),\ldots,\sigma_n(x_{\pi(n)}))$.
\end{theorem}
\begin{proof}
Let $f\in P_k^n$ be an arbitrary function and assume $Ess(f)=X_n$. 

First, we shall prove that 
\[f(x_1,\ldots,x_n)\simeq_{imp} f(x_{\pi(1)},\ldots,x_{\pi(n)})\] {and} \[f(x_1,\ldots,x_n)\simeq_{sub} f(x_{\pi(1)},\ldots,x_{\pi(n)}).\]
Let $g=f(x_{\pi(1)},\ldots,x_{\pi(n)})$.
Clearly, if $n\leq 1$ then $f\simeq_{imp} g$ and $f\simeq_{sub} g$.
Assume that if $n<p$ then $f\simeq_{imp} g$ and $f\simeq_{sub} g$ for some natural number $p$, $p\geq 1$.

 Let us suppose $n=p$. Let $x_i\in Ess(f)$  be an arbitrary  essential variable in $f$ and let $c\in Z_k$ be an arbitrary constant from $Z_k$.  Then we have 
\[f(x_i=c)(x_1,\ldots,x_{i-1},x_{i+1},\ldots,x_p)=\] 
 \[=g(x_{\pi^{-1}(i)}=c)(x_{\pi^{-1}(1)},\ldots,x_{\pi^{-1}({i-1})},x_{\pi^{-1}({i+1})},\ldots,x_{\pi^{-1}(p)}).\]
  Our inductive assumption implies
$f(x_i=c)\simeq_{imp}g(x_{\pi(i)}=c)$
 {and}  $sub_m(f(x_i=c))=sub_m(g(x_{\pi(i)}=c))$ for all $x_i\in X_n$, $m\in\{1,\ldots,p-1\}$ and $c\in Z_k$. Hence
$f\simeq_{imp}g
\ \mbox{and}\  f\simeq_{sub}g$.

Second, let us prove that 
\[f(x_1,\ldots,x_n)\simeq_{imp} f(\sigma_1(x_{1}),\ldots,\sigma_n(x_{n}))\]
 {and}
\[f(x_1,\ldots,x_n)\simeq_{sub} f(\sigma_1(x_{1}),\ldots,\sigma_n(x_{n})).\]
Let $h=f(\sigma_1(x_{1}),\ldots,\sigma_n(x_{n}))$. Then we have 
\[f(a_1,\ldots,a_n)=h(\sigma_1^{-1}(a_1),\ldots,\sigma_n^{-1}(a_n)).\]
Hence, if $(i_1\ldots i_r, a_{i_1}\ldots a_{i_r}c)\in Imp(f)$ then $(i_1\ldots i_r, \sigma_{i_1}^{-1}(a_{i_1})\ldots \sigma_{i_r}^{-1}(a_{i_r})c)\in Imp(h)$ for some $r$, $1\leq r\leq n$.
Since $\sigma_i$ is a permutation of $Z_k$ for $i=1,\ldots,n$ it follows that 
$f\simeq_{imp}h$. By similar arguments it follows that 
$f\simeq_{sub}h$.
\end{proof}

\begin{corollary}\label{c5.4} (i)  $GE_k^n\leq IM_k^n$; \quad (ii) $GE_k^n\leq SB_k^n$; \quad
(iii) $GE_k^n\leq SP_k^n$. 
\end{corollary}

\section{Classification of Boolean Functions}\label{sec6}

In this section we compare a collection  of subgroups of   RAG    with the groups of transformations preserving the relations $\simeq_{imp}$, $\simeq_{sub}$ and $\simeq_{sep}$  and to obtain estimations for  the number of equivalence classes, and for the cardinalities of these classes in the algebra of Boolean functions. 
Our results are based on Proposition \ref{p5.1}, Theorem \ref{t5.2} and Theorem \ref{t5.3}. Thus we have 
\begin{equation}\label{eq3}
GE_2^n\leq IM_2^n,\quad GE_2^n\leq SB_2^n, \quad LG_2^n\not\leq SP_2^n \quad\mbox{and}\quad LF_2^n\not\leq SP_2^n.
\end{equation}

These relationships determine the places of the groups $IM_2^n$, $SB_2^n$ and $SP_2^n$ with respect to the subgroups of RAG. Figure \ref{f4} shows the location of  these groups together with the subgroups of RAG.

M. Harrison \cite{har2} and R. Lechner  \cite{lech}   counted  the number of equivalence classes and the cardinalities of the classes  under some transformation subgroups of RAG  for Boolean functions of 3 and 4 variables. 

The relations  (\ref{eq3}) show that if we have the values of $t(GE_2^n)$ then we can count the numbers $t(IM_2^n)$, $t(SB_2^n)$ and $t(SP_2^n)$ because the equivalence classes under these transformation groups are union of equivalence classes under $GE_2^n$ and hence we have $t(IM_2^n)\leq t(GE_2^n)$ and $t(SB_2^n)\leq t(GE_2^n)$. Moreover, if we know the factor-set $P_2^n/_{\simeq_{gen}}$ of representative functions  under $\simeq_{gen}$ then  we can effectively calculate the sets  $P_2^n/_{\simeq_{imp}}$,  $P_2^n/_{\simeq_{sub}}$ and  $P_2^n/_{\simeq_{sep}}$ because of  $P_2^n/_{\simeq_{imp}}\subseteq P_2^n/_{\simeq_{gen}}$ and  $P_2^n/_{\simeq_{sub}}\subseteq P_2^n/_{\simeq_{gen}}$.

The next theorem allows us to count the number $imp(f)$ of the implementations of any function $f$ by a recursive procedure. Such a procedure is realized and its execution is used when calculating the number of the implementations and classifying  the functions under the equivalence $\simeq_{imp}$.
\begin{theorem} \label{t35} Let $f\in P_2^n$ be a boolean function. The number of all  implementations in $f$   is determined as follows:
\[imp(f) = \left\{\begin{array}{ccc}
            1 \  &\  if \  &\ ess(f)=0 \\
            &&\\
          2 & if\  & ess(f)=1\\
           &&\\
           \sum_{x\in Ess(f)}[imp(f(x=0)) + imp(f(x=1))] & if\ & ess(f)\geq 2.
           \end{array}
           \right.
\]
\end{theorem}
\begin{proof}
We shall proceed by induction on $n=ess(f)$ - the number of essential variables in $f$. 
The lemma is clear if $ess(f)=0$.  If $f$ depends essentially on one variable $x_1$, then there is a unique BDD of $f$ with one non-terminal node which has two outcoming edges. These edges together with the labels of the corresponding terminal nodes form the set $Imp(f)$ of all implementations of $f$, i.e. $imp(f)=2$.

Let us assume that 
\[imp(f)= \sum_{i=1}^n[imp(f(x_i=0)) + imp(f(x_i=1))]\] 
if $n< s$ for some natural number $s$, $1\leq s$.

Next, let us consider a function $f$ with $ess(f)=s$. Without loss of generality, assume that $Ess(f)=\{x_1,\ldots,x_n\}$ with $n=s$.
Since $x_i\in Ess(f)$ for $i=1,\ldots,n$ it follows that $f(x_i=0)\neq f(x_i=1)$ and there exist BDDs of $f$ whose label of the first non-terminal node is $x_i$. Let $D_f$ be a such BDD of $f$ and let $(ij_2\ldots j_m,c_1c_2\ldots c_mc)\in Imp(f)$ with $m\leq n$. Hence 
\[(j_2\ldots j_m,c_2\ldots c_mc)\in Imp(g)\] where $g=f(x_i=c_1)$. On the other side it is clear that if 
$(j_2\ldots j_m,d_2\ldots d_md)\in Imp(g)$ then $(ij_2\ldots j_m,c_1d_2\ldots d_md)\in Imp(f)$. Consequently, there is an one-to-one mapping between the set of implementations of $f$ with first variable $x_i$ and first edge labelled by $c_1$, and $Imp(g)$, which completes the proof.
\end{proof}
We  also develop recursive algorithms to count $sub_m(f)$ and $sep_m(f)$ for  $f\in P_2^n$, presented below.

Table \ref{tb4} shows the number of equivalence classes under the equivalence relations induced by the transformation groups $G_2^n$, $IM_2^n$, $SB_2^n$ and $SP_2^n$ for $n=1,2,3,4$.  M. Harrison  found from applying Polya's counting theorem (see \cite{har2}) the numbers $t(G_2^5)$  and $t(G_2^6)$,  which are upper bounds of $t(IM_2^n)$, $t(SB_2^n)$ and $t(SP_2^n)$ for $n=5,6$. 

Figure \ref{f4} and Table \ref{tb3} show that  for the algebra $P_2^3$ there are only 14 different generic equivalent classes, 13 imp-classes, 11 sub-classes and 5 sep-classes. Hence three mappings that converts each generic class into an imp-class, into a sub-class and into a sep-class are required.  Each generic class is a different row of Table \ref{tb3}. For example, the generic class \textnumero  12 (as it is numbered in Table VIII, \cite{lech}) is presented by 10-th row of Table \ref{tb3}. It consists of 8 functions obtained by complementing function $f$ and/or permuting and/or complementing input variables in all possible ways, where $f=x_1x_2^0x_3\oplus x_1x_2x_3^0\oplus x_2x_3$.  This generic class \textnumero  12 is included in imp-class \textnumero 9, sub-class \textnumero  8 and sep-class \textnumero  5 which shows that $imp(f)=36$, $sub(f)=12$ and $sep(f)=7$. The average cardinalities of equivalence classes and complexities of functions are also shown in the last row of Table \ref{tb3}.

Table \ref{tb5} shows the $sep$-classes of boolean functions depending on at most five variables. Note that there are $2^{32}=4294967296$ functions in $P_2^5$. All calculations were performed on a computer with two Intel Xeon E5/2.3 GHz CPUs. The execution with total exhaustion  took 244 hours.

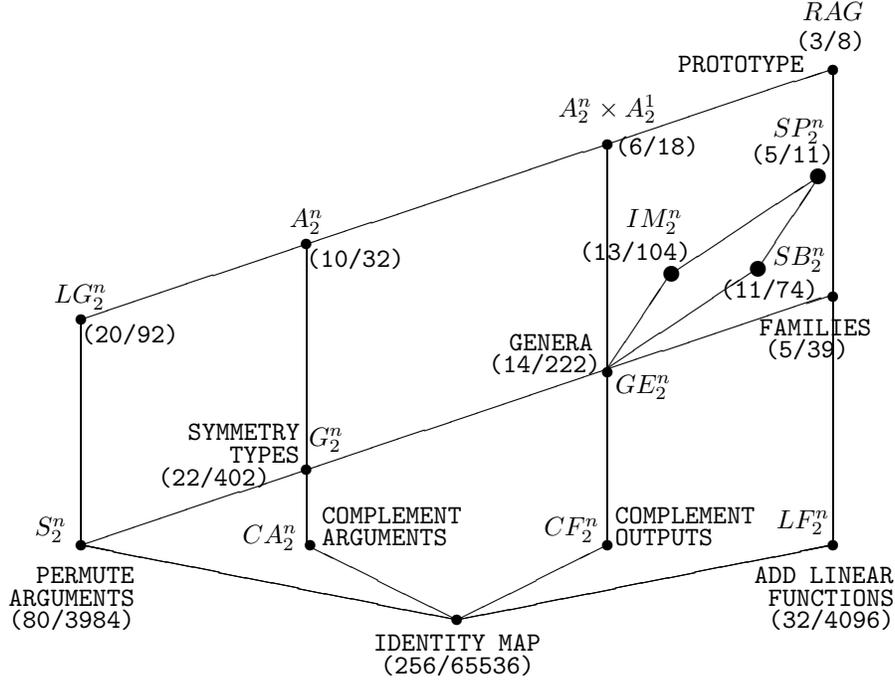
\begin{figure}%[h]
\centering
\unitlength 1mm % = 2.845pt
\linethickness{0.4pt}
\begin{picture}(120.006,88.001)(0,0)

\put(60,8){\line(5,1){50}}
\put(60,8){\line(-5,1){50}}

\put(60,8){\line(2,1){20}}
\put(60,8){\line(-2,1){20}}

\put(60,8){\line(5,1){50}}
\put(60,8){\line(-5,1){50}}

\put(10,18){\line(0,1){30}}
\put(110,18){\line(0,1){63.5}}

\put(40,18){\line(0,1){40}}
\put(80,18){\line(0,1){53.5}}

\put(10,18){\line(3,1){100}}

\put(10,48){\line(3,1){100}}

\put(10,18){\circle*{1.5}}
\put(6, 20){\makebox(0,0)[cc]{$S_2^n$}}
\put(17, 14){\makebox(0,0)[rr]{\ttfamily PERMUTE}}
\put(17, 11){\makebox(0,0)[rr]{\ttfamily  ARGUMENTS}}
\put(17, 8){\makebox(0,0)[rr]{\ttfamily (80/3984)}}

\put(40.5,18){\circle*{1.5}}
\put(42, 22){\makebox(0,0)[ll]{\ttfamily COMPLEMENT}}
\put(42, 19){\makebox(0,0)[ll]{\ttfamily ARGUMENTS}}
\put(31.5, 19){\makebox(0,0)[ll]{$CA_2^n$}}

\put(10,48){\circle*{1.5}}
\put(10,51){\makebox(0,0)[cc]{\ttfamily $LG_2^n$}}
\put(10,46){\makebox(0,0)[ll]{\ttfamily (20/92)}}

\put(40,58){\circle*{1.5}}
\put(40,61){\makebox(0,0)[cc]{\ttfamily $A_2^n$}}
\put(40,56){\makebox(0,0)[ll]{\ttfamily (10/32)}}

\put(80,71.25){\circle*{1.5}}
\put(80,76){\makebox(0,0)[cc]{\ttfamily $A_2^n\times A_2^1$}}
\put(81,71){\makebox(0,0)[ll]{\ttfamily (6/18)}}

\put(110,81.25){\circle*{1.5}}
\put(110,89){\makebox(0,0)[cc]{$RAG$}}
\put(105,85){\makebox(0,0)[ll]{\ttfamily (3/8)}}
\put(106, 82){\makebox(0,0)[rr]{\ttfamily PROTOTYPE}}

\put(60,8){\circle*{1.5}}
\put(60, 5){\makebox(0,0)[cc]{\ttfamily IDENTITY MAP}}
\put(60, 2){\makebox(0,0)[cc]{\ttfamily (256/65536)}}

\put(40,28){\circle*{1.5}}
\put(39, 33){\makebox(0,0)[rr]{\ttfamily SYMMETRY}}
\put(39, 30){\makebox(0,0)[rr]{\ttfamily  TYPES}}
\put(45, 32){\makebox(0,0)[rr]{ $G_2^n$}}
\put (35,27){\makebox(0,0)[rr]{\ttfamily (22/402)}}

\put(80,18){\circle*{1.5}}
\put(81, 22){\makebox(0,0)[ll]{\ttfamily COMPLEMENT}}
\put(81, 19){\makebox(0,0)[ll]{\ttfamily OUTPUTS}}
%\put(80.5, 17){\makebox(0,0)[ll]{\ttfamily (128/32768)}}
\put(79, 20){\makebox(0,0)[rr]{ $CF_2^n$}}

\put(80, 41){\circle*{1.5}}
\put(78, 45){\makebox(0,0)[rr]{\ttfamily GENERA}}
\put (79,42){\makebox(0,0)[rr]{\ttfamily (14/222)}}
\put (88.5,39){\makebox(0,0)[rr]{$GE_2^n$}}

\put(110,51){\circle*{1.5}}
\put(115, 47){\makebox(0,0)[rr]{\ttfamily FAMILIES}}
\put (112,44){\makebox(0,0)[rr]{\ttfamily (5/39)}}

\put(110,18){\circle*{1.5}}
\put(102.5,21){\makebox(0,0)[ll]{\ttfamily  $LF_2^n$}}
\put(118, 14){\makebox(0,0)[rr]{\ttfamily ADD LINEAR}}
\put(118, 11){\makebox(0,0)[rr]{\ttfamily   FUNCTIONS}}
\put (118,8){\makebox(0,0)[rr]{\ttfamily (32/4096)}}

\put(80,41.5){\line(2,3){8}}
%\put(90.5,57){\line(3,1){28}}
\put(80,41.5){\line(3,2){20}}

\put(100,54.75){\line(2,3){8}}

\put(88.5,54){\line(3,2){20}}
\put(108,67){\circle*{2}}
\put(88.5,54){\circle*{2}}
\put(90,61){\makebox(0,0)[rr]{\ttfamily  $IM_2^n$}}
%\put (87,52){\makebox(0,0)[rr]{\ttfamily PAIRS}}
\put (91,57){\makebox(0,0)[rr]{\ttfamily (13/104)}}

\put(100,54.75){\circle*{2}}

\put(102,56){\makebox(0,0)[ll]{\ttfamily  $SB_2^n$}}
\put (95,52){\makebox(0,0)[ll]{\ttfamily (11/74)}}

\put(102,73){\makebox(0,0)[ll]{\ttfamily  $SP_2^n$}}
\put(99,70){\makebox(0,0)[ll]{\ttfamily  (5/11)}}

\end{picture}

%\vspace{.7cm}

\caption{Transformation groups in $P_2^n$ ($n=3/n=4$)}\label{f4}

\end{figure}

\begin{table}%[h]
\centering
\caption{Number of classes under $symmetry$ $type$, $\simeq_{imp}$, $\simeq_{sub}$ and $\simeq_{sep}$}\label{tb4}

\begin{tabular}{rrrrr}
$n$ & $t(G_2^n)$ & $t(IM_2^n)$ &$t(SB_2^n)$ &$t(SP_2^n)$  \\ \hline\hline
1&3&2&2&2\\ %\hline
2&6&4&4&3 \\
3&22&13&11&5 \\
4&402&104&74&11 \\
5&1\ 228\ 158&1606&{$<$ 1228158}& 38 \\
6&400\ 507\ 806\ 843\ 728&\multicolumn{3}{c}{$<$ 400\ 507\ 806\ 843\ 728} \\
\hline\hline
\end{tabular}
\end{table}

\begin{sidewaystable}[p]\centering

\vspace{4in}

\caption{Classification of  $P_2^3$ under $\simeq_{sep}$, $\simeq_{sub}$, $\simeq_{imp}$ and genus.} \label{tb3}

%\small
\begin{tabular}{||c|c|c||c|c|c||c|c|c||c|c||c||} \hline
\hline
sep- 	 & 	$sep(f)$	 & 	func.    & 	sub-  & 	$sub(f)$	 & func. & imp-	 & 	$imp(f)$	 & 	func.& Generic & func. & 	representative		\\
class 	 & 	         	 & per 	  & class 		 & 	       	 & 	per 	 & 	class		 & 	         	 & per &class \cite{lech}& per & 	function	 $f$       \\
\textnumero &&class&\textnumero &&class&\textnumero &&class&\textnumero  &class&\\
  \hline\hline
1	 & 	0	 & 	2		 & 	1	 & 	1	 & 	2	 & 	1	 & 	1	 & 	2&1&2  & 	$0$	 \\ \hline
2	 & 	1	 & 	6		 & 	2	 & 	3	 & 	6	 & 	2	 & 	2	 & 	6&9&6& 	$x_1$	 	\\ \hline
\multirow{2}{*}{3}	 & \multirow{2}{*}{3}	 & 	\multirow{2}{*}{30}		 & 	3	 & 	5	 & 	24	 & 	3	 & 	6	 & 	24&3&24& 	$x_1x_2$	 	 \\ \cline{4-12}
	 & 		 & 			 & 	4	 & 	7	 & 	6	 & 	4	 & 	8	 & 	6&10&6& 	$x_1\oplus x_2$	\\ \hline
4	 & 	6	 & 	24	 & 	5	 & 	11	 & 	24	 & 	5	 & 	28	 & 	24&13&24& 	$x_1\oplus x_1x_3\oplus$ 		\\ 
  	 &  	 & 	 	  	 & 	 & 	 	 & 	 	 & 	 	 & 		 & 	 &&& 	$ x_2x_3$	 \\
 \hline
\multirow{9}{*}{5}	 & \multirow{9}{*}{7}	 & \multirow{9}{*}{194}	 	 & 	\multirow{2}{*}{6}	 & 	\multirow{2}{*}{9}	 & \multirow{2}{*}{64}	 & 	6	 & 	21	 & 	16&2&16 & 	$x_1x_2x_3$	\\ \cline{7-12}
	 & 		 & 		 & 		 & 		 & 		 & 	7	 & 	23	 & 	48&6&48& 	$x_1x_2^0x_3^0\oplus x_1$		 \\ \cline{4-12}
	 & 		 & 			 & 	7	 & 	12	 & 	48	 & 	8	 & 	30	 & 	48&7&48 & 	$x_1x_2^0x_3^0\oplus $	 \\  & 		 & 			 	 & 		 & 	 	 & 		 & 		 & 		 & 	& & & 	$x_2x_3$ \\ \cline{4-12}
	 & 		 & 			 & 	8	 & 	12	 & 	8	 & 	\multirow{2}{*}{}	 & 	\multirow{2}{*}{}	 & 	\multirow{2}{*}{}&12&8& 	$x_1x_2^0x_3\oplus$	 \\ 	 & 		 & 		  	 & 	 	 & 	 	 & 	 	 & 	\multirow{2}{*}{9}	 & 	\multirow{2}{*}{36}	 & 	\multirow{2}{*}{16}&& & 	$ x_1x_2x_3^0 \oplus x_2x_3$	 \\ \cline{4-6} \cline{10-12}
	 & 		 & 			 & 	\multirow{3}{*}{9}	 & 	\multirow{3}{*}{15}	 & 	\multirow{3}{*}{26}	 & 		 & 		 & &5&8& 	$x_1^0x_2x_3\oplus x_1x_2^0x_3^0$			\\ \cline{7-12}  
	 & 		 & 		 & 		 & 		 & 		 & 	10	 & 	42	 & 	16&8&16& 	$x_1x_2^0x_3\oplus $	 \\ 

 & 		 & 			 & 	 	 & 		 & 		 & 		 & 		 & 		 & 	&& 	$ x_1x_2x_3^0\oplus x_1^0x_2x_3$ \\ 		 
	% \cline{4-5}
	  \cline{7-12}
	 & 		 & 			 & 		 & 		 & 		 & 	11	 & 	48	 & 	2& 11&2& 	$x_1\oplus x_2\oplus x_3$	\\ \cline{4-12}
	 & 		 & 		 	 & 	10	 & 	13	 & 	24	 & 	12	 & 	32	 & 	24&14&24 & 	$x_1\oplus x_2x_3$	 \\ \cline{4-12}
	 & 		 & 			 & 	11	 & 	13	 & 	24	 & 	13	 & 	33	 & 	24&4&24 & 	$x_1x_2^0x_3\oplus x_1x_2x_3^0$	\\ \hline
aver.&6.2&51.2&&10.6&23.3& &26.0&19.7&&18.3&\\ \hline\hline
\end{tabular}
\end{sidewaystable}
 
\begin {table}
\centering
\caption{Classes in  $P_2^5$ under $\simeq_{sep}$ }\label{tb5}
%\small
\noindent
\begin{tabular}{||c|c|c|c|c|c|c|c||} \hline\hline

sep- & \multirow{3}{*}{\rotatebox{00}{$sep_5(f)$}} & \multirow{3}{*}{\rotatebox{0}{$sep_4(f)$}} & \multirow{3}{*}{\rotatebox{0}{$sep_3(f)$}} & \multirow{3}{*}{\rotatebox{0}{$sep_2(f)$}} & \multirow{3}{*}{\rotatebox{0}{$sep_1(f)$}} & \multirow{3}{*}{\rotatebox{90}{$sep(f)$}} & functions \\
class &  &   &   &   &   &   & per   \\
\textnumero & &   &   &   &   &   &  class \\ \hline

1 & 0 & 0 & 0 & 0 & 0 & 0 & 2\\ \hline
2 & 0 & 0 & 0 & 0 & 1 & 1 & 10\\ \hline
3 & 0 & 0 & 0 & 1 & 2 & 3 & 100\\ \hline
4 & 0 & 0 & 1 & 2 & 3 & 6 & 240\\ \hline
5 & 0 & 0 & 1 & 3 & 3 & 7 & 1940\\ \hline
6 & 0 & 1 & 2 & 5 & 4 & 12 & 1920\\ \hline
7 & 0 & 1 & 3 & 4 & 4 & 12 & 2400\\ \hline
8 & 0 & 1 & 3 & 5 & 4 & 13 & 8160\\ \hline
9 & 0 & 1 & 4 & 4 & 4 & 13 & 120\\ \hline
10 & 0 & 1 & 4 & 5 & 4 & 14 & 8400\\ \hline
11 & 0 & 1 & 4 & 6 & 4 & 15 & 301970\\ \hline
12 & 1 & 2 & 7 & 9 & 5 & 24 & 20480\\ \hline
13 & 1 & 3 & 5 & 7 & 5 & 21 & 3840\\ \hline
14 & 1 & 3 & 5 & 8 & 5 & 22 & 9600\\ \hline
15 & 1 & 3 & 6 & 6 & 5 & 21 & 1920\\ \hline
16 & 1 & 3 & 6 & 7 & 5 & 22 & 1920\\ \hline
17 & 1 & 3 & 6 & 8 & 5 & 23 & 38400\\ \hline
18 & 1 & 3 & 7 & 7 & 5 & 23 & 1920\\ \hline
19 & 1 & 3 & 7 & 8 & 5 & 24 & 38400\\ \hline
20 & 1 & 3 & 7 & 9 & 5 & 25 & 130560\\ \hline
21 & 1 & 4 & 6 & 6 & 5 & 22 & 3000\\ \hline
22 & 1 & 4 & 7 & 7 & 5 & 24 & 34720\\ \hline
23 & 1 & 4 & 7 & 8 & 5 & 25 & 177120\\ \hline
24 & 1 & 4 & 7 & 9 & 5 & 26 & 274560\\ \hline
25 & 1 & 4 & 8 & 7 & 5 & 25 & 7680\\ \hline
26 & 1 & 4 & 8 & 8 & 5 & 26 & 274560\\ \hline
27 & 1 & 4 & 8 & 9 & 5 & 27 & 1847280\\ \hline
28 & 1 & 5 & 7 & 9 & 5 & 27 & 81920\\ \hline
29 & 1 & 5 & 8 & 8 & 5 & 27 & 600\\ \hline
30 & 1 & 5 & 8 & 9 & 5 & 28 & 1013760\\ \hline
31 & 1 & 5 & 8 & 10 & 5 & 29 & 38400\\ \hline
32 & 1 & 5 & 9 & 7 & 5 & 27 & 1200\\ \hline
33 & 1 & 5 & 9 & 8 & 5 & 28 & 449040\\ \hline
34 & 1 & 5 & 9 & 9 & 5 & 29 & 4093200\\ \hline
35 & 1 & 5 & 9 & 10 & 5 & 30 & 5443200\\ \hline
36 & 1 & 5 & 10 & 8 & 5 & 29 & 13680\\ \hline
37 & 1 & 5 & 10 & 9 & 5 & 30 & 5826160\\ \hline
38 & 1 & 5 & 10 & 10 & 5 & 31 & 4274814914\\ \hline\hline
\end{tabular}
\end{table}

\end{document}